\documentclass[prx,aps,twocolumn,notitlepage,superscriptaddress,showpacs,nofootinbib]{revtex4-2}

\usepackage{enumerate,appendix}
\usepackage{qcircuit}
\usepackage{amsmath, amsthm, amssymb,commath}
\usepackage{color,calc,graphicx}
\usepackage[usenames,dvipsnames,svgnames,table,cmyk,hyperref]{xcolor}
\usepackage[colorlinks]{hyperref}
\usepackage{optidef}
\hypersetup{
	colorlinks = true,
	urlcolor = {blue},
	citecolor = {blue},
	linkcolor= {blue}
}

\usepackage{graphicx}
\usepackage{amsmath}
\usepackage{latexsym}
\usepackage{bbm}

\usepackage[charter,cal=cmcal,sfscaled=false]{mathdesign}
\usepackage{booktabs}
\usepackage{multirow} 
\usepackage{dcolumn}
\usepackage{mathrsfs}
\usepackage{csvsimple-l3}

\def \be {\begin{equation}}
\def \ee {\end{equation}}

\newcommand{\ket}[1]{|#1\rangle}
\newcommand{\bra}[1]{\langle#1|}

\def\>{\rangle}
\def\<{\langle}

\theoremstyle{remark}	

\newtheorem{definition}{Definition}
\newtheorem{theorem}{Theorem}
\newtheorem{lemma}[theorem]{Lemma}

\newtheorem{remark}[theorem]{Remark}

\newtheorem{example}{Example}

\newcommand{\ky}[1]{{#1}}

\newcommand{\ed}[1]{{#1}}

\usepackage{algorithm}
\usepackage[caption=false,labelformat=simple]{subfig}	

\newcommand{\sM}{{\cal M}}
\newcommand{\sN}{{\cal N}}

\usepackage{bm}
\newcommand{\vA}{\bm{A}}
\newcommand{\vB}{\bm{B}}

\newcommand{\vE}{\bm{E}}
\newcommand{\vF}{\bm{F}}

\newcommand{\vH}{\bm{H}}

\newcommand{\vs}{\bm{s}}

\newcommand{\vLb}{\bm{\Lambda}}

\usepackage{mathbbol}
\newcommand{\br}{\mathbb{r}}

\usepackage{caption,graphicx,newfloat} 
\DeclareFloatingEnvironment[
  fileext=lob,
  listname={List of Boxes},
  name=Box,
  placement=htp,
]{BOX}

\usepackage{caption}
\usepackage{makecell} 

\usepackage{nccmath} 

\usepackage{pdfpages}
\makeatletter
\AtBeginDocument{\let\LS@rot\@undefined}
\makeatother
\usepackage{pgffor}

\begin{document}

\title{Degenerate quantum erasure decoding}

\author{Kao-Yueh Kuo}
\email{kywukuo@saturn.yzu.edu.tw}
\affiliation{School of Mathematical and Physical Sciences, University of Sheffield, Sheffield, S3 7RH, United Kingdom}
\affiliation{Department of Electrical Engineering, Yuan Ze University, Taoyuan City 320315, Taiwan}
\author{Yingkai Ouyang}
\email{y.ouyang@sheffield.ac.uk}
\affiliation{School of Mathematical and Physical Sciences, University of Sheffield, Sheffield, S3 7RH, United Kingdom}

\begin{abstract}
Erasures are the primary type of errors in physical systems dominated by leakage errors. 
While quantum error correction (QEC) using stabilizer codes can combat erasure errors, it remains unknown which constructions achieve capacity performance.
If such codes exist, decoders with linear runtime in the code length are also desired.
In this paper, we present erasure capacity-achieving quantum codes under maximum-likelihood decoding (MLD), though MLD requires cubic runtime in the code length. 
For QEC, using an accurate decoder with the shortest possible runtime will minimize the degradation of quantum information while awaiting the decoder's decision.  
To address this, we propose belief propagation (BP) decoders that run in linear time and exploit error degeneracy in stabilizer codes, achieving capacity or near-capacity performance for a broad class of codes, including bicycle codes, product codes, and topological codes.
We furthermore explore the potential of our BP decoders to handle mixed erasure and depolarizing errors, and also local deletion errors via concatenation with permutation invariant codes.  
\end{abstract}

\maketitle

\section*{Introduction}

Erasure conversion \cite{Blu+24,SST+23,MLP+23,WKPT22,Con+22,  KHV+23,TWB+23,LHH+24, KCB23, LGZ+08} is a recent technique that converts leakage errors into erasures, building on earlier frameworks where the computational Hilbert space is embedded in a larger physical space \cite{GBP97,BDS97,DZ13}.
In systems dominated by leakage errors, such as ultracold Rydberg atoms \cite{Blu+24,SST+23,MLP+23,WKPT22,Con+22}, superconducting circuits \cite{KHV+23,TWB+23,LHH+24}, 
trapped ions \cite{KCB23}, 
or photonic systems \cite{LGZ+08},
erasure conversion ensures that the majority of errors are erasures, by paying necessary physical resource or time for the conversion. 
In the context of quantum error correction (QEC), the promised structure of errors affords potentially better error thresholds and decoding efficiency compared to unstructured errors.
We aim to explore optimal QEC performance in erasure-dominated systems.

The rate of a QEC code is the ratio of the number of encoded qubits $k$ to the number of physical qubits $n$.
The quantum erasure channel capacity, $C(p) = 1-2p$ \cite{BDS97,DZ13}, represents the maximum rate ($k/n$) for which a decoder can almost surely correct typical erasures, where $p$ is the independent erasure probability per qubit. This capacity quantifies the highest rate of reliable quantum information transmission over a noisy erasure channel.
A central question in quantum information theory is: how closely can practical QEC codes approach $C(p)$ with low-complexity decoders?

This question is particularly pertinent for quantum stabilizer codes with code parameters $[[n,k]]$ and low-weight commutative Pauli operators,
known as quantum low-density parity-check (QLDPC) codes. 
Two outstanding open problems in this area are: (1) identifying stabilizer code constructions that achieve capacity $r = 1 - 2p$ for arbitrary rates $r = k/n$, and (2) developing decoders that enable these constructions to achieve near-capacity performance with linear runtime complexity.
This paper solves both of these problems.

Previously, only two-dimensional (2D) topological codes, such as surface or toric codes, were known to achieve erasure capacity \cite{SBD09,Ohz12} with linear decoding time \cite{DZ20}, but their code rate vanishes ($r\to 0$) as $n$ increases. 

A recent breakthrough by IBM significantly reduce qubit overhead ($1-\frac{k}{n}$) compared to surface codes, using non-zero rate QLDPC codes---specifically bivariate bicycle codes---highlighting their non-vanishing rate and potential for nearly 2D layouts \cite{Bra+24}. 
Bicycle-type codes with hundreds of qubits achieve comparable performance to surface codes with thousands, protecting a similar number of information qubits, offering a 10~times reduction in the number of physical qubits.

A further significant reduction in qubit overhead can be achieved using higher-rate codes, with long code lengths ensuring the correction of a large fraction of errors, which we include in this work.
First, we show that the bicycle construction \cite{MMM04} achieves capacity for moderate-to-high rate codes on the quantum erasure channel.
Second, for low-rate (but non-zero rate) cases, product constructions with flexible rate, such as lifted-product (LP) codes \cite{PK22}, also achieve erasure capacity.
Third, to facilitate practical applications, we provide linear-time decoders universal for these codes, also supporting vanishing-rate 2D topological codes.

For erasures, the task of decoding (inferring the error) reduces to solving a system of linear equations (cf.~\eqref{eq:H4}--\eqref{eq:dec_era_eqs}). 
The optimal decoder is the maximum-likelihood decoder (MLD), which is often implemented using a Gaussian decoder. This approach employs Gaussian elimination with a complexity of $O(n^3)$, or $O(n^{2.373})$ \cite{Wil12,LeG14}.

In classical erasure error-correction, a binary bit-flipping belief propagation (Flip-BP$_2$) decoder is preferred over MLD for classical LDPC codes due to its linear runtime $O(n)$ and near-capacity performance \cite{SS96,LMSS01_era,Di+02,RU01_e,MP03,Tan+04,Sho06}. 
The success of BP is evident in 5G networks \cite{Pop+18,Zai+18,Aro+20}.
However, naively applying Flip-BP$_2$ to QLDPC codes is problematic due to inevitable small stopping sets in the quantum case 
(see \ed{Supplementary Note~3 for more details}). 
Small stopping sets 
obstruct BP convergence \cite{Ric03,Di+02,OUVZ02,RU01_e,MP03,Tan+04,Sho06} and are avoided in the classical code design \cite{RU01_e,MP03,Tan+04,Sho06}.

In the quantum case, various techniques have been proposed for decoding QLDPC codes under erasures.
First, 
    the peeling decoder was proposed to decode 2D topological codes, including peeling on toric or surface codes (in linear time) \cite{DZ20} or extended peeling on color codes (with additional complexity) \cite{AS19,SS23,LMS20}. 
    (The peeling decoder includes a spanning-tree step \cite{CLRS09}, followed by a peeling step resembling a Flip-BP$_2$ in \cite{LMSS01_era}, as pointed in \cite{CLLD24}.)
 Second, for product codes, the pruned peeling (improved over peeling decoder) and a Vertical-Horizontal (VH) post-processing were proposed \cite{CLLD24}. 
    The ``pruned peeling + VH'' has complexity $O(n^2)$, or $O(n^{1.5})$ for a probabilistic version. 
 Third, modifications of minimum-weight perfect matching (MWPM) were showed good erasure decoding accuracy
  for 2D codes \cite{KCB23,SJC+23}, though without linear time \cite{Edm65}. 
Recently, BP was improved with guided decimation \cite{GYP24}; and cluster decomposition (extending the VH technique) was proposed in \cite{YGP24} to balance accuracy and complexity between MLD and VH. Both works \cite{GYP24,YGP24} focused on product codes.

In this paper, we study a broad class of QLDPC codes, including 2D topological codes \cite{BK98,Kit97_03,DKLP02,BM07} and their variants such as XZZX codes \cite{KDP11,THD12,ATBFB21}, as well as bicycle codes \cite{MMM04}, hypergraph-product (HP) codes \cite{TZ09_14}, generalized HP (GHP) codes \cite{KP13}, and LP codes \cite{PK22}.
These are CSS-type codes \cite{CS96,Steane96}, except for the XZZX codes.
General QLDPC codes admit low-complexity quantum circuits and offer the benefits of (1) enabling error-corrected quantum computation with constant encoding rates \cite{Got14c,BE21} and (2) enhancing quantum network communication rates \cite{LK24}.

In our approach, we define \emph{erasure bounded-distance decoding} (eBDD) (Def.~\ref{def:eBDD}) as a metric for evaluating decoding accuracy.
We estimate the error threshold $p_\text{th}$ of a code family by checking whether a constant correction radius $t/n$ exists as $n$ increases, where $t$ is the expected number of corrected erasures. If so, we set $p_\text{th}=t/n$.
For clarity, we may denote $p^\text{(BP)}_\text{th}$ for BP and $p^\text{(MLD)}_\text{th}$ for MLD if they differ. 
We emphasize that the threshold values in this paper apply to independent erasures; other error types, such as depolarizing errors, yield different thresholds.

Quantum decoding relies on exploiting degeneracy, enabling a decoder to converge to any (degenerate) error equivalent to the occurred error up to a stabilizer, resulting in remarkable accuracy \cite{PC08,PK21,KL22}.
Our constructed codes achieve capacity performance with $r = 1 - 2p^\text{(MLD)}_\text{th}$ using the Gaussian decoder (MLD), which always finds a solution likely to be a degenerate error, at the cost of $O(n^3)$ complexity.
To achieve linear runtime $O(n)$, we first propose a gradient-descent Flip-BP$_2$ (GD Flip-BP$_2$) to address stopping sets via a simple GD scheme 
(which is supported by the symmetry of degenerate errors; see \cite{PC08} or \ed{Supplementary Note~4 for more discussions}). 
Next, we introduce a soft binary BP with memory effects (MBP$_2$), its adaptive version (AMBP$_2$), and their quaternary counterparts (A)MBP$_4$. 
A soft BP decoder mimics a soft GD optimization, providing strong convergence to degenerate solutions. 
For our constructed codes, we achieve near-capacity performance with $r = 1-2.5p^\text{(BP)}_\text{th}$ using (A)MBP or GD Flip-BP$_2$ for non-zero rate codes, with AMBP providing very strong accuracy at low error rates. In addition, AMBP$_4$ achieves the optimal threshold of 0.5 for vanishing-rate 2D topological codes.

The success of (A)MBP stems from multiple factors.
First, we apply message softization \eqref{eq:soft} to have non-extreme, and non-zero BP message values, enabling updates within stopping sets.
Second, we link BP to gradient descent to determine a proper update step-size and, when necessary, adjust the prior error distribution to ensure sufficient update gradients, further mitigating harmful stopping sets. 
Third, we extend binary decoders to quaternary ones, leveraging correlations between the $X$-portion and \mbox{$Z$-portion} of a QLDPC check matrix to prevent BP from being halted by stopping sets in one portion.
Moreover, inspired by Hopfield nets \cite{Hop84, HT85, HT86}, 
(A)MBP uses fixed inhibition (see \ed{Supplementary Note~5 (part~b)}), 
adding memory effects to BP to improve convergence.
Additionally, we employ a group random schedule for message updates to improve convergence, similar to a serial schedule \cite{KL20}, 
while processing multiple messages in parallel (see \ed{Supplementary Note~8}).
Finally, we propose a strategy to make AMBP as efficient as MBP to facilitate practical applications.

We also explore the potential of our decoders to manage deletions of untracked particle losses (with the assistance of permutation-invariant (PI) codes \cite{Yin14,Yin21}) or mixed erasure and depolarizing errors.

\section*{Results} \label{sec:results}

\begin{figure}
    \centering \includegraphics[width=0.485\textwidth]{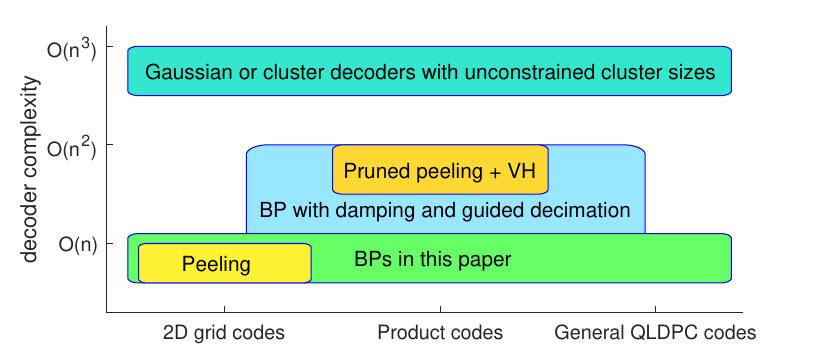}
    \caption{Comparison of decoders.
    The vertical axis indicates complexity, while the horizontal axis shows code types with increasing connectivity from left to right.
    BPs in this paper perform accurately across various code types, similar to the Gaussian decoder, with a slight threshold gap but with linear complexity. 
    The achieved thresholds are shown in Fig.~\ref{fig:bnd_logx}.
    } \label{fig:var_decs} 
\end{figure}

\begin{figure}[t!] 
    \centering \includegraphics[width=0.5\textwidth]{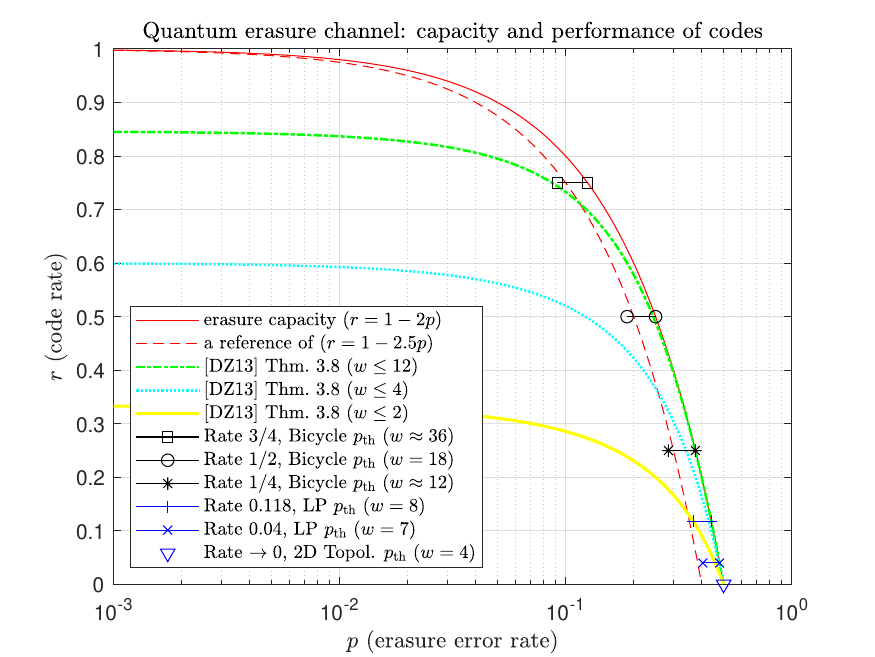}
    \caption{
    Quantum erasure capacity and thresholds of codes.
    For 2D topological codes, BP and MLD share a threshold of $p_\text{th} = 0.5$. 
    For other codes, a small accuracy gap is indicated by two points: $p^\text{(BP)}_\text{th}$ (left) and $p^\text{(MLD)}_\text{th}$ (right).
    The constructed codes achieve capacity with MLD, $r = 1 - 2 p^\text{(MLD)}_\text{th}$, and near-capacity with BP, $r = 1 - 2.5 p^\text{(BP)}_\text{th}$.
    We summarize our decoders and tabulate their accuracy figures in Table~\ref{tb:cmp_dec}.
    For bicycle codes, BP thresholds are achieved using GD~Flip-BP$_2$ or (A)MBP$_q$.
    For LP codes, BP thresholds are achieved using AMBP$_q$, and GD Flip-BP$_2$ shows similar accuracy for rate 0.118 LP codes.
    For 2D topological codes, a threshold of 0.5 is achieved using AMBP$_4$.
    %
    Higher-rate codes require larger stabilizer weight $w$, consistent with the trend from \cite{DZ13} (plotted as curves labeled [DZ13]). 
    } \label{fig:bnd_logx} 
\end{figure}

\begin{table*}
\caption{
    Comparison of linear-time and MLD decoders for quantum erasure errors on data qubits.
    } \label{tb:cmp_dec}
\centering 
{\footnotesize 
\resizebox{\textwidth}{!}{ 
\begin{tabular}{|c|c|c|c|c|c|}
\hline
decoder & type & \makecell{optimized\\ complexity} & \makecell{worst case\\ complexity} & tested codes & comment \\
\hline
\makecell{Peeling \cite{DZ20}} & \makecell{spanning-tree\\(MLD for\\ 2D codes)} & \makecell{$O(n)$\\ (toric/surface)} & \makecell{$O(nT_{\max})$} & \makecell{2D gird (toric/surface) \cite{DZ20}; color codes\\ need additional complexity \cite{AS19,SS23,LMS20}} & \makecell{spanning-tree step is linear-time \cite{CLRS09}; 
peeling\\ step is iterative, but guaranteed to finish within\\ $T_{\max}=O(1)$ iterations for toric/surface codes.
}\\
\hline
\makecell{Pruned peeling\\ + VH \cite{CLLD24}} & \makecell{spanning tree\\ + local GE 
} & \makecell{$O(n^{1.5})$\\(if use \cite{Wie86})} & $O(n^2)$ & HP & 
\makecell{Peeling needs VH to decode HP codes\\ but underperforms AMBP}\\ 
\hline
\makecell{BP with damping and\\
guided decimation \cite{GYP24}} & \makecell{soft BP + guided\\ decimation} & \makecell{$O(n)$ (reduced\\ accuracy)} & $O(n^2 T_{\max})$ & HP and GHP & 
\makecell{outperform peeling+VH, although\\ not as accurate as AMBP}\\ 
\hline
\makecell{Peeling + cluster\\ decomposition  \cite{YGP24}} & \makecell{spanning tree\\ + tuneable GE} 
& \makecell{$O(n)$ (reduced\\ accuracy)} & \makecell{$O(n^3)$\\ (MLD)} & HP and GHP & 
\makecell{flexible accuracy/complexity tradeoff;\\ can beat BPs but need complexity $O(n^3)$}\\
\hline
\hline
\makecell{GD Flip-BP$_2$\\ ({\bf this paper},\\ {\bf Algorithms~\ref{alg:GD_BP2}})} & bit-flipping BP & $O(n\log\log n)$ & $O(nT_{\max})$ 
& \makecell{LP (Fig.~\ref{fig:decs_on_LP118}), bicycle (Fig.~\ref{fig:bic_wt_up}),\\ HP and GHP (Fig.~\ref{fig:GHP})} 
& \makecell{performs well on bicycle and rate 0.118 LP\\ codes; extremely low complexity} \\
\hline
\makecell{(A)MBP$_q$ ($q=2$ or 4)\\ ({\bf this paper},\\ {\bf Algorithms~\ref{alg:MBP2}--\ref{alg:BP_era}})} & soft BP & $O(n\log\log n)$ & $O(nT_{\max})$ 
& \makecell{LP (Fig.~\ref{fig:LP118}), toric and XZZX (Fig.~\ref{fig:2D_lin}),\\ HP and GHP (Fig.~\ref{fig:GHP}),\\ bicycle (Supplementary Note~7 (part~c))} 
& \makecell{accurate and universal for QLDPC codes; also\\ 
decode mixed errors with high accuracy (Fig.~\ref{fig:Eras_Dep})} \\
\hline
\makecell{Gaussian elimination\\ (GE) (see \eqref{eq:dec_sys_eqs}--\eqref{eq:dec_era_eqs})} & MLD & \makecell{$O(n^{2.373})$\\ (see \cite{Wil12,LeG14})} & $O(n^3)$ & \makecell{we provide the accuracy achieved by\\ Gaussian decoder for all tested codes} & \makecell{MLD is difficult to simulate at low error rates,\\ but eBDD provides good curve extrapolation} \\
\hline
\end{tabular}
}\\[5pt] 
{\bf Upper portion: decoders in the literature; Lower portion: our decoders and GE.}
GHP codes are considered as special LP codes 
    \cite{PK22,GYP24,YGP24}.\\ 
We construct flexible-rate and flexible-length LP codes from random quasi-cyclic matrices and bicycle codes with random generator vectors. 
High-rate bicycle codes ($r=3/4$) are particularly suitable to work with PI codes for local deletion errors (cf.~Fig.~\ref{fig:Convert}). 
Bicycle codes are generalized for fault-tolerant error correction on IBM hardware implementations, emphasizing their non-zero code rate and potential for nearly 2D layouts \cite{Bra+24}.
\ed{Figures comparing our decoders with those in \cite{CLLD24,GYP24,YGP24} are provided in Supplementary Note~7 (part~a).}
}
\end{table*}

Figure~\ref{fig:var_decs} compares the decoding complexity of various decoders. 
BPs developed in this paper have accuracy comparable to the Gaussian decoder (MLD), but with linear complexity.
The thresholds achieved by BP and MLD are analysed and compared with the channel capacity in Fig.~\ref{fig:bnd_logx}.
Physical fidelity decays exponentially over time (characterized by the decoherence time \cite{Blu+22,MLP+23}).
We plot the erasure error rate $p$ on a log scale to depict the exponential effort to maintain coherence.
In Fig.~\ref{fig:bnd_logx}, except for the 2D topological codes, each code family is represented by two data points, indicating the thresholds achieved by BP and MLD.
There is a small gap between the accuracy of BP (linear-time decoder) and MLD (non-linear time decoder).
The each code family has nearly fixed row and column weights to ensure linear BP complexity 
(see \ed{Supplementary Note~7 (part~b)}).

In Fig.~\ref{fig:bnd_logx}, we also plot rate upper bounds estimated by \cite[Theorem~3.8]{DZ13}, which depend on stabilizer weight $w$. 
While not necessarily tight, these bounds align with our findings in that higher code rates require higher $w$.

To obtain a precise estimate of the rate currently achieved via BP at erasure probability $p$, we replot Fig.~\ref{fig:bnd_logx} on a linear scale in 
\ed{Supplementary Note~7 (part~f)}. 
By interpolation, the rate achieved by BP is $r = 1 - 2.5p - 2p^2 + 6p^3$.

In Table~\ref{tb:cmp_dec}, we compare our results with other linear-time or MLD decoders in the literature \cite{DZ20,CLLD24,GYP24,YGP24}.

\subsection*{Quantum erasure channel} \label{sec:eras_ch}

Erasures occur when certain qubits are lost, with a detector identifying their locations. 
A lost qubit is represented by an ``erasure state'' $\ket{2}$, orthogonal to the computational basis states $\ket{0}$ and $\ket{1}$. 
Thus, we consider \cite{BDS97,DZ13}:
\begin{definition}
A quantum erasure channel independently erases each qubit $\rho_1$ with probability $p \le 0.5$, leaving it unchanged otherwise, 
resulting in the following process:
    \begin{equation} \label{eq:eras_ch}
    \textstyle
    \rho_1 \mapsto (1-p)\rho_1 + p\ket{2}\bra{2},
    \end{equation}
with the location of the erased qubit known.
\end{definition}
Assume that an $[[n,k]]$ stabilizer code is used.
For decoding, we insert an ancilla qubit at each lost qubit's location, creating a noisy codeword state of length $n$ for stabilizer measurements. After measurements, each ancilla qubit collapses randomly, encountering Pauli $I, X, Y, \text{ or } Z$ with equal probability $1/4$.
Thus, an erased qubit (with its location known) encounters the following error before decoding:
    \begin{equation} \label{eq:like_dep_ch}
    \textstyle
    \rho_1 \mapsto (1-p) \rho_1 + p\frac{I}{2},
    \end{equation}
where
    $\frac{I}{2}= \sum_{W\in\{I,X,Y,Z\}}\frac{1}{4}W\rho_1W^\dagger$.
The error process \eqref{eq:like_dep_ch} is similar to the depolarizing channel, but with a key distinction: we know which qubits are completely depolarized.

Let $\br\subseteq\{1,2,\dots,n\}$ be the set of \emph{erased coordinates}. 
If $j\in\br$, then qubit $j$ is erased.
The complement, ${\bar\br}={\{1,2,\dots,n\}\setminus \br}$, contains the \emph{non-erased coordinates}.

Let $E= {E_1\otimes \dots\otimes E_n}$ be the $n$-fold Pauli error (up to global phase) after measurements, 
where $E_j=I,X,Y, \text{ or } Z$ with equal probability $\frac{1}{4}$ if $j\in\br$, and $E_j=I$ if $j\in\bar{\br}$. 

A stabilizer decoder infers $E$ from $\br$ and the \emph{syndrome}~of~$E$ obtained through measurements, as detailed below.

\subsection*{Erasure decoding problem in stabilizer codes} \label{sec:era_dec}

We refer to \cite{CRSS98,GotPhD} for background on stabilizer codes. 
For decoding, we represent the measured stabilizers $H^{(1)}, \dots, H^{(n-k)}$ as a \emph{(quaternary) Pauli check matrix}
	\begin{equation} \label{eq:H4}
	H = \left[\begin{smallmatrix}
	H^{(1)}\\ \vdots\\ H^{(n-k)}\\
	\end{smallmatrix}\right]
	\in \{I,X,Y,Z\}^{(n-k)\times n}.
	\end{equation}
To have a linear decoding problem, we map $I,X,Y,Z$ to $(0|0), (1|0), (1|1), (0|1)$, respectively. 
Conventionally, this defines a map $\varphi: \{I,X,Y,Z\}^n \to \{0,1\}^{2n}$, and we have a $(n-k)\times 2n$ \,\emph{binary check matrix} 
    \begin{align} 
      \vH 
      &=\left[\begin{smallmatrix}
        \varphi(H^{(1)})\\ 
        \vdots\\ 
        \varphi(H^{(n-k)})
	\end{smallmatrix}\right]
	 =
	\left[\begin{smallmatrix}
	\vH_1 \\ 
        \vdots\\ 
        \vH_{n-k}
	\end{smallmatrix}\right]
	 =
	\left[\begin{smallmatrix}
	\vH_{1,1}	&\dots	&\vH_{1,2n}	\\ 
	\vdots		&\ddots	&\vdots		\\
	\vH_{n-k,1}	&\dots	&\vH_{n-k,2n}
	\end{smallmatrix}\right] \label{eq:H2}\\
	&= [\vH^X\,|\,\vH^Z]
	 =
		\left[\begin{smallmatrix}
		\vH^X_{1,1}	&\dots	&\vH^X_{1,n}	\\ 
		\vdots		&\ddots	&\vdots		\\
		\vH^X_{n-k,1}&\dots	&\vH^X_{n-k,n}
		\end{smallmatrix}\middle|
		\begin{smallmatrix}
		\vH^Z_{1,1}	&\dots	&\vH^Z_{1,n}	\\ 
		\vdots		&\ddots	&\vdots		\\
		\vH^Z_{n-k,1}&\dots	&\vH^Z_{n-k,n}
		\end{smallmatrix} \right]. \label{eq:HXHZ}
    \end{align} 
Consider the operation under ${\rm mod}\ 2$ for the binary field $\mathbb F_2  = \{0,1\}$.
For two binary matrices $\vA$ and $\vB$ of $2n$ columns, define a bilinear form (called \emph{symplectic inner product})
    $$
    \langle \vA,\vB \rangle = \vA\vLb \vB^T,
    $$
where 
  $\vLb = \left[ {\bm{0}_n\atop \bm{1}_n} {\bm{1}_n\atop \bm{0}_n} \right]$
and $(\cdot)^T$ is the matrix transpose.
Two operators $E,F\in\{I,X,Y,Z\}^n$ commute if the corresponding $\vE,\vF$ satisfy $\langle \vE,\vF \rangle = 0$, and anticommute if $\langle \vE,\vF \rangle = 1$. 
We denote zero and identity matrices by $\bm{0}$ and $\bm{1}$, respectively, using the subscript $n$ to specify dimension $n \times n$.
Notice that  $\langle\vH,\vH\rangle=\bm{0}$. 
The row-space of $\vH$ is a linear code 
$C = \text{Row}(\vH) \subset\mathbb F_2^{2n}$, which is symplectic self-orthogonal. 
That is, $C\subseteq C^\perp$, where $C^\perp$ is the symplectic dual of $C$ in $\mathbb F_2^n$.

For any stabilizer codes, we have this {\bf degeneracy}:
for an error $\vE\in\{0,1\}^{2n}$, its \emph{logical coset} is $\vE+C$; any errors $\tilde{\vE} \in \vE+C$ are \emph{degenerate errors} of $\vE$ with an equivalent effect on the code space.
Degenerate errors share the same syndrome $\vs = \langle \vE, \vH \rangle = \langle \tilde{\vE}, \vH \rangle$, but syndrome matched errors form a larger error set called \emph{error coset} ($\vE+C^\perp$) of syndrome $\vs$.

\begin{definition} \label{def:main_QMLD}
(Quantum erasure decoding problem) 
Given a binary matrix $\vH\in\{0,1\}^{m\times 2n}$, $m\ge n-k$, erasure information $\br\subseteq\{1,2,\dots,n\}$, and a syndrome $\vs\in\{0,1\}^m$, where $\vs=\langle \vE,\vH \rangle $ for some $\vE\in\{0,1\}^{2n}$ drawn from the erasure channel, a decoder is asked to find an $\hat\vE\in\{0,1\}^{2n}$ such that $\hat{\vE} \in \vE+C$ with probability as high as possible.
\end{definition}

Given $\br$ and $\vs$, an erasure- and syndrome-matched solution satisfies a system of linear equations:
\begin{equation} \label{eq:dec_sys_eqs}
    \begin{cases}
    ~ (\vH\vLb) \vE^T 
        = [\vH^Z|\vH^X] (\vE^X|\vE^Z)^T = \vs^T, \\
    ~ (\vE^X_j|\vE^Z_j) 
        = (0|0) ~~ \forall\, j\in\bar{\br}.
    \end{cases}
\end{equation} 
\begin{definition}
A solution $\hat{\vE}\in\{0,1\}^{2n}$ satisfying \eqref{eq:dec_sys_eqs} is called a \emph{feasible solution} for the erasure decoding.
\end{definition}
Solving \eqref{eq:dec_sys_eqs} can be simplified by solving 
\begin{equation} \label{eq:dec_era_eqs}
    [\vH^{Z}_\br|\vH^{X}_\br](\vE^{X}_\br|\vE^{Z}_\br)^T = \vs^T,
\end{equation}
where $A_\br$ denotes the submatrix of a matrix $A$ restricted to the columns specified by $\br$. 
The cardinality  ${|\br|\approx np \propto n}$. 
A~solution to \eqref{eq:dec_sys_eqs} or \eqref{eq:dec_era_eqs} can be obtained by Gaussian elimination with complexity $O(n^3)$.
Since the syndrome $\vs$ is produced from the physical error $\vE$, there always exists at least one feasible solution (e.g., $\vE$ itself) to the linear equations in \eqref{eq:dec_sys_eqs}. However, the solution may not be unique, especially at high erasure error rates.
The uniqueness of the solution depends on $|\br|$ and the code distance $d$ 
(discussed in detail in Supplementary Note~2).
An $[[n,k]]$ code with distance $d$ is called an $[[n,k,d]]$ code.

Other error models, such as those with syndrome errors or circuit-level noise, can lead to non-unique solutions even at low physical error rates; these cases are beyond the scope of this paper.
	We assume noiseless syndrome extraction. 
	Similar deductions (as those in Lemma~\ref{lm:set_eq}, Theorem~\ref{th:MLD}, 
	and Supplementary Note~1) 
	still apply under syndrome errors or circuit-level noise, given appropriate check matrices and logical generators. Decoding still proceeds based on the check matrix constructed as in \cite{KL25,KL24a}, but the problem becomes more challenging, as it may yield many degenerate solutions even at low physical error rates.

Paralleling \cite[Lemma~1]{DZ20}, we present the following lemma and theorem, proved by group theory and probability theory in 
Supplementary Note~1.

\begin{lemma} \label{lm:set_eq}
If two distinct logical cosets both contain feasible solutions,
then the number of feasible solutions in both cosets is equal, and they both are maximum likelihood logical cosets.   
\end{lemma}

\begin{theorem} \label{th:MLD}
For a stabilizer code with a check matrix $\vH\in\{0,1\}^{m\times 2n}$, if erasures $\br\subseteq\{1,2,\dots,n\}$ and syndrome $\vs\in\{0,1\}^m$ are given for decoding, 
then any feasible solution to \eqref{eq:dec_sys_eqs} is an optimal solution for the decoding problem in Def.~\ref{def:main_QMLD}.
A decoder always outputting a feasible solution is an MLD.
\end{theorem}

Obtaining a feasible solution is the best a decoder can do for erasure decoding, although it does not always guarantee successful decoding, especially when $|\br|$ is large.

\begin{example} \label{ex:4_1_code}
Consider a $[[4,1]]$ code with this check matrix
	\begin{align}
	\vH = [\vH^X\,|\,\vH^Z] &=
        \left[\begin{smallmatrix}
	1&0&0&0\\
	0&1&0&1\\
	0&0&1&1\\
	\end{smallmatrix}\middle|
	\begin{smallmatrix}
	0&0&1&0\\
	0&1&0&1\\
	1&0&0&1\\
	\end{smallmatrix}\right], 
    \label{eq:4_1_code_H2} 
	\end{align}
or equivalently, its Pauli check matrix
	\begin{equation} \label{eq:4_1_code_H4}
	H = \left[\begin{smallmatrix}
	X&I&Z&I\\
	I&Y&I&Y\\
	Z&I&X&Y\\
	\end{smallmatrix}\right].
	\end{equation}
There are $2^{2k} = 4$ logical cosets for a given syndrome $\vs$.
When $\vs = (0\,1\,0)$, the corresponding 4 logical cosets, in Paulis, are: 
    \begin{equation*}
    \left[\begin{smallmatrix}
    I&Z&I&I&\leftarrow\\
    X&Z&Z&I\\
    I&X&I&Y&\leftarrow\\
    X&X&Z&Y\\
    Z&Z&X&Y\\
    Y&Z&Y&Y\\
    Z&X&X&I\\
    Y&X&Y&I\\
    \end{smallmatrix}\right], 
    \quad
    \left[\begin{smallmatrix}
    I&I&Z&X\\
    X&I&I&X\\
    I&Y&Z&Z\\
    X&Y&I&Z\\
    Z&I&Y&Z\\
    Y&I&X&Z\\
    Z&Y&Y&X\\
    Y&Y&X&X\\
    \end{smallmatrix}\right], 
    \quad
    \left[\begin{smallmatrix}
    I&I&Z&Z\\
    X&I&I&Z\\
    I&Y&Z&X\\
    X&Y&I&X\\
    Z&I&Y&X\\
    Y&I&X&X\\
    Z&Y&Y&Z\\
    Y&Y&X&Z\\
    \end{smallmatrix}\right], 
    \quad
    \left[\begin{smallmatrix}
    I&Z&I&Y&\leftarrow\\
    X&Z&Z&Y\\
    I&X&I&I&\leftarrow\\
    X&X&Z&I\\
    Z&Z&X&I\\
    Y&Z&Y&I\\
    Z&X&X&Y\\
    Y&X&Y&Y\\
    \end{smallmatrix}\right],
    \end{equation*} 
where left arrows indicate the errors to be discussed.

When only the second qubit is erased, we easily identify a unique solution $IZII$. MLD always succeeds in this error case.

When the 2nd and 4th qubits are erased, i.e., $\br = \{2,4\}$, we have $E_j=I$ for $j\in\bar{\br}=\{1,3\}$. 
As indicated by the left arrows above,
the occurred error is a Pauli operator in either 
	\begin{equation} \label{eq:opt_cosets}
	\left[\begin{smallmatrix}
	I&Z&I&I\\
	I&X&I&Y\\
	\end{smallmatrix}\right]
	\quad\text{or}\quad
	\left[\begin{smallmatrix}
	I&Z&I&Y\\
	I&X&I&I\\
	\end{smallmatrix}\right].
	\end{equation}
An MLD will output a solution in one of the cosets, but the actual error may lie in the other coset, leading to a logical error rate of $1/2$ in this error case.

\ed{Using the above error cases, we illustrate the coset structure in Supplementary Note~1 (see S~Fig.~1 therein).}
\end{example}

A good code can resolve many erasures with $|\br| = np$ at large $p$, and we aim to resolve them in linear time.

\subsection*{Bit-Flipping BP algorithm}  \label{sec:flip_BP}

Given \eqref{eq:dec_sys_eqs}, BP can approximate the solution with high accuracy, by running a message passing on the Tanner graph induced by the code's check matrix. 
For classical BP, see \cite{Gal63,Tan81,Pea88,Mac99,KFL01}. 
For quantum decoding, we can perform binary decoding based on a binary check matrix $\vH\in\{0,1\}^{m\times 2n}$, or a quaternary decoding based on a Pauli check matrix $H\in\{I,X,Y,Z\}^{m\times n}$. 
For example, matrices in \eqref{eq:4_1_code_H2} and \eqref{eq:4_1_code_H4} induce the two Tanner graphs shown in Fig.~\ref{fig:4_1_code}.

For a Tanner graph, let $\sN(i)$ be the set of neighbouring variable node indices of check node $i$.
Let $\sM(j)$ be the set of neighbouring check node indices of variable node~$j$.
For example, for a Tanner graph induced by an $\vH\in\{0,1\}^{m\times 2n}$,
	\begin{align*}
	\sN(i) &= \{j\in\{1,2,\dots,2n\}: \vH_{ij}=1\},\\
	\sM(j) &= \{i\in\{1,2,\dots,m\}: \vH_{ij}=1\}.
	\end{align*}

\begin{figure}
	\hspace*{-4mm}\centering\includegraphics[width=0.54\textwidth]{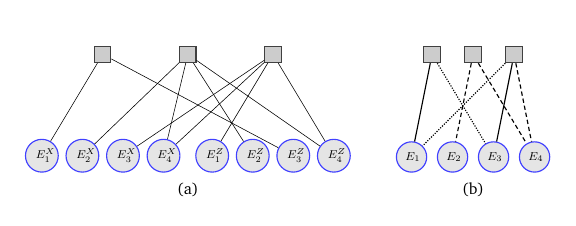}
	\caption{
        (a) Tanner graph induced by the binary matrix $\vH$ in \eqref{eq:4_1_code_H2}.
        (b) Tanner graph induced by the Pauli matrix $H$ in \eqref{eq:4_1_code_H4}, with edge types $X$ (solid line), $Y$ (dashed line), and $Z$ (dotted line).
	} \label{fig:4_1_code}		
\end{figure}

For BP erasure decoding, important code properties include \emph{stopping sets} and \emph{girth}. 
For details, see Supplementary Note~3. 
However, due to the degeneracy of stabilizer codes, the solutions are generally not unique. 
Specifically, the number of feasible solutions in a logical coset (Lemma~\ref{lm:set_eq}) increases with the number of erasures (cf.~Example~\ref{ex:4_1_code}). 
When decoding a code, we aim for the decoder to combat an erasure rate $p$ as high as possible, which leads to many erasures and, consequently, many feasible solutions in the logical coset where the actual error is located.
In particular, we observe strong symmetry of degenerate errors (see \cite{PC08} or more details in 
Supplementary Note~4).
This motivates us to use heuristics to converge to any feasible solution.
Thus, we propose a bit-flipping algorithm with a gradient descent (GD) step, referred to as GD Flip-BP$_2$, as outlined in Algorithm~\ref{alg:GD_BP2}. 
(This algorithm behaves more like a greedy approach to minimize a discrete function~\eqref{eq:obj}, but can be linked to GD optimization for a differentiable function~\eqref{eq:soft_obj}.)

We explain Algorithm~\ref{alg:GD_BP2}.
The binary check matrix has $2n$ columns. 
Given ${ \br\subseteq\{1,2,\dots,n\} }$, 
we have
	$${ \br_2=\br\cup(\br+n) },$$ 
where 
	$\br+n = \{j+n:j\in\br\}$.
	%
The GD step performs when no updates are performed during an iteration.
This step tries to minimize an objective function defined as the sum of the numbers of unresolved bits in each check, expressed as 
    \begin{equation} \label{eq:obj}
    \textstyle \sum_{i=1}^m |\sN(i)\cap\br_2|.
    \end{equation}

\begin{figure}
\begin{algorithm}[H]
\begin{flushleft}
\caption{: GD Flip-BP$_2$} \label{alg:GD_BP2}
	{\bf Input}:\\ 
	\quad A binary matrix $\vH\in\{0,1\}^{m\times 2n}$, a~syndrome $\vs\in\{0,1\}^m$,\\ 
	\quad erased coordinates $\br_2\subseteq\{1,2,\dots,2n\}$, 
	and an~integer $T_{\max}>0$. \\
    \quad ($\vH=[\vB^Z|\vB^X]$ from a binary check matrix $\vB=[\vB^X|\vB^Z]$.) \\[2pt]

	{\bf Initialization.} 
	\begin{itemize}
	\item Let $\Gamma_j = 0$ if $j\in\br_2$, and $\Gamma_j = +1$ if $j\in\bar{\br}_2$.
	\end{itemize}

	{\bf Iterative Update.} 
	\begin{itemize}
	\item 
		For $i\in\{1,2,\dots,m\}$: 
        Compute $\Delta_i = |\sN(i)\cap\br_2|$;
        Then if $\Delta_i=1$,
		let $j$ be the single element in $\sN(i)\cap{\br_2}$ and set
		\begin{align*}
		\Gamma_j = (-1)^{\vs_i}\prod_{j'\in\sN(i)\setminus \{j\}} \Gamma_{j'}.
		\end{align*}

        \item (GD Step) If no $\Gamma_j$ is updated in the previous sub-step:
        Find a $j\in\br_2$ such that the corresponding column in $\vH_{\br_2}$ has a maximum column weight. Set $\Gamma_j = -1$.
  
        \item 
        If any $\Gamma_j$ is updated during this iteration, remove $j$ from $\br_2$.
	\end{itemize}

	{\bf Judgement.} 
	\begin{itemize}
	\item 
        If $\br_2 = \emptyset$: 
        For all $j=1,\dots,2n$,
        if $\Gamma_j=+1$, set $\hat\vE_j = 0$, 
        and if $\Gamma_j=-1$, set $\hat\vE_j = 1$. 
        Generate $\hat{\vs} = \langle\hat{\vE},\vB\rangle = \hat{\vE}\vH^T$.
            \begin{itemize}
                \item if $\hat{\vs} = \vs$: return ``CONVERGE'';
                \item else: return ``GD\_FAIL''.
            \end{itemize}
	\item 
		Else: if the maximum number of iterations $T_{\max}$ is reached, halt and return ``FAIL''.
	\item 
		Repeat from the step of Iterative Update.	
	\end{itemize}
\end{flushleft}
\end{algorithm}
\end{figure}

\begin{figure}
\begin{algorithm}[H]
\begin{flushleft}
\caption{: MBP$_2$ (with an optional soft GD step)} \label{alg:MBP2}
	{\bf Input}:\\ 
	\quad A binary matrix $\vH\in\{0,1\}^{m\times 2n}$, a~syndrome $\vs\in\{0,1\}^m$,\\ 
    \quad an integer $T_{\max}>0$, a real $\alpha>0$, and initial LLRs \\
	\quad $\Lambda_j = \ln(p_j^{(0)}/p_j^{(1)})$ for $j=1,2,\dots,2n$ (cf.~\eqref{eq:main_p2_init}). 
    ~\\ 
    \quad ($\vH=[\vB^Z|\vB^X]$ from a binary check matrix $\vB=[\vB^X|\vB^Z]$.) 
    \\[5pt]
    \quad GD parameters: an integer $T_\text{GD}>0$ and a real $|\Lambda_\text{GD}|>0$.
    \\[5pt]

	{\bf Initialization.} 
	\begin{itemize}
	\item For $j\in\{1,2,\dots,2n\}$ and $i\in{\cal M}(j)$: 
		Let $\Gamma_{j\to i} = \text{soft}(\Lambda_j)$.
	\end{itemize}
		 
	{\bf Check-Node Computation.} 
    \begin{itemize}
    \item For $i\in\{1,2,\dots,m\}$: Let 
        $\Delta_i = (-1)^{\vs_i}\boxplus_{j\in\sN(i)} \Gamma_{j\to i}$.
    \item For $i\in\{1,2,\dots,m\}$ and $j\in\sN(i)$: Let 
        $\Delta_{i\to j} = \Delta_i \boxminus \Gamma_{j\to i}$.
    \end{itemize}

	{\bf Variable-Node computation.} 
	\begin{itemize}
	\item For $j\in\{1,2,\dots,2n\}$: Let 
        $\Gamma_j = \Lambda_j + (1/\alpha)\sum_{i\in\sM(j)} \Delta_{i\to j}$. 
    \item For $j\in\{1,2,\dots,2n\}$ and $i\in{\cal M}(j)$: Compute
        \begin{align}
        \Gamma_{j\to i} &= \Gamma_j - \Delta_{i\to j}, \label{eq:fix_inh}\\
        \Gamma_{j\to i} &= \text{soft}(\Gamma_{j\to i}). \notag
        \end{align}

    \item (Optional soft GD Step) When the number of iterations is a multiple of $T_\text{GD}$: 
    For $j\in\{1,2,\dots,2n\}$,
    \begin{equation}\label{eq:MBP2_GD}
    \text{if $|\Gamma_j|<|\Lambda_\text{GD}|$, then set $\Lambda_j = \text{sign}(\Gamma_j)\,|\Lambda_\text{GD}|$.}
    \end{equation}

    \end{itemize}

	{\bf Judgement.} 
	\begin{itemize}
        \item For all $j=1,\dots,2n$, if $\Gamma_j\ge 0$, set $\hat\vE_j = 0$, and if $\Gamma_j<0$, set $\hat\vE_j = 1$. 
        Generate $\hat{\vs} = \langle\hat{\vE},\vB\rangle = \hat{\vE}\vH^T$.
            \begin{itemize}
                \item if $\hat{\vs} = \vs$: return ``CONVERGE'';
                \item else if the maximum number of iterations $T_{\max}$ is reached: halt and return ``FAIL''.
                \item else: repeat from the Check-Node Computation step.	
            \end{itemize}
	\end{itemize}
\end{flushleft}
\end{algorithm}
\end{figure}

\subsection*{MBP$_2$ and its adaptive version (AMBP$_2$)} \label{sec:MBP2}

A soft-valued decoder improves decoding accuracy. 
A bit could be $0,1$ or erased (denoted by $?$). The values ${0, ?, 1}$ are mapped as ${+1, 0, -1}$ in GD Flip-BP$_2$ for computation.  
Once a decision is made, $?$ becomes a bit value 0 or 1 and remains fixed, though it may be incorrect.
Generally, we consider a soft-valued decoder, where the messages are represented as soft values, allowing them to adjust their signs and magnitudes during iterations, improving decision accuracy.

We describe this decoder generalization. We have the initial probabilities $p_j^{(0)}=\Pr(\vE_j=0)$ and $p_j^{(1)}=\Pr(\vE_j=1)$ for $j=1,2,\dots,2n$, where
    \begin{equation} \label{eq:main_p2_init}
    (p_j^{(0)},p_j^{(1)}) = (p_{j+n}^{(0)},p_{j+n}^{(1)}) =
        \begin{cases}
        (\tfrac{1}{2},\tfrac{1}{2}) \quad&\text{if $j\in\br$,}\\
        (1,0) \quad&\text{if $j\in\bar\br$.}
        \end{cases}
    \end{equation}
The initial \emph{log-likelihood ratios} (LLRs) are $\Lambda_j = \ln(p_j^{(0)}/p_j^{(1)})$ for $j=1,2,\dots,2n$. 
The (hard) values ${0, ?, 1}$ of $\vE_j$ map to ${+\infty, 0, -\infty}$ as $\Lambda_j$.
We will update these LLRs, and denote the updated LLRs as $\Gamma_j$, $j=1,2,\dots,2n$, to estimate~$\hat \vE_j$. 

We restrict variable-to-check messages from taking values in ${-\infty, 0, \infty}$, avoiding representing the three (hard) values ${0, ?, 1}$.  The function achieves our purpose:
    \begin{equation} \label{eq:soft}
    \text{soft}(x) = 
        \medmath{\begin{cases}
        \text{sign}(x)\,\text{LLR\_MAX}, &\text{if $|x|>\text{LLR\_MAX}$},\\
        \text{sign}(x)\,\text{LLR\_MIN}, &\text{if $|x|<\text{LLR\_MIN}$},\\
        x, &\text{otherwise}.
        \end{cases}}
    \end{equation}
This maintains LLR numerical stability and enables updates even within stopping sets. Suggested values for \text{LLR\_MAX} and \text{LLR\_MIN} are provided in 
Supplementary Note~9.

\begin{remark} \label{rm:init_no_soft} 
We do not apply the function $\text{soft}(\cdot)$ in \eqref{eq:soft} to initial LLRs $\Lambda_j$. 
This ensures $\Gamma_j=\Lambda_j=+\infty$ for $j\in\bar\br_2$, making our BP output always erasure-matched, though convergence (syndrome matching, $\hat\vs=\vs$) is not guaranteed.
\end{remark}

We will update LLRs to meet the target syndrome.
For two LLRs $\Gamma_1$ and $\Gamma_2$, following Gallager \cite{Gal63}, we define
    \begin{equation}
    \Gamma_1\boxplus\Gamma_2 := \text{sign}(\Gamma_1\Gamma_2) ~ f^{-1} \left( f(|\Gamma_1|) + f(|\Gamma_2|) \right),
    \end{equation}
where $|\cdot|$ is absolute value,  
    $f(x)=\ln\frac{e^x+1}{e^x-1}=\ln(\coth(\frac{x}{2}))$, 
and notice that $f^{-1}(x)=f(x)$.
An equivalent expression is 
    $\Gamma_1\boxplus\Gamma_2 = 2\tanh^{-1} \left( \tanh(\frac{\Gamma_1}{2}) \tanh(\frac{\Gamma_2}{2}) \right)$ \cite{KFL01},
and in general, 
    \begin{equation} \label{eq:bsum} 
	\overset{w}{\underset{j=1}{\boxplus}} \Gamma_j := 2\tanh^{-1}\left( {\textstyle \prod_{j=1}^w} \tanh\tfrac{\Gamma_j}{2} \right). 
    \end{equation}
To resolve all erasures (minimizing \eqref{eq:obj}) while satisfying the target syndrome, we define a check-satisfaction reward function for maximization:
    \begin{equation} \label{eq:soft_obj}
    \textstyle S(\Gamma_1,\Gamma_2,\dots,\Gamma_{2n}) := \sum_{i=1}^m ~ (-1)^{\vs_i} \underset{j\in\sN(i)}{\boxplus} \Gamma_j.
    \end{equation}

Algorithm~\ref{alg:MBP2} presents MBP$_2$, performing a soft message-passing to approximate a solution to \eqref{eq:dec_sys_eqs} by updating LLRs.
Specifically, it approximates $\Gamma_j$ to
    $\ln\medmath{\frac
    {\Pr(\vE_j=0~|~\text{syndrome is } \vs)}
    {\Pr(\vE_j=1~|~\text{syndrome is } \vs)}}$ 
by maximizing \eqref{eq:soft_obj} 
\ed{(see \eqref{eq:llr_tanh_a}--\eqref{eq:a_vs_g})}.

The check-to-variable messages $\Delta_{i\to j}$ are normalized by some $\alpha>0$, like classical normalized BP \cite{CF02a,CDE+05}.
Typically, ${\alpha \ge 1}$ is chosen to suppress overestimated messages. Here, we also consider ${\alpha < 1}$ to enlarge the BP step-size to search more possible solutions. 
Larger step-size may cause divergence.
Inspired by Hopfield nets \cite{Hop84, HT85, HT86}, MBP uses fixed inhibition \eqref{eq:fix_inh}, adding memory effects to BP to improve convergence even with ${\alpha<1}$ \cite{KL22} 
(see \ed{Supplementary Note~5 (part~b)}).

We can select a proper $\alpha$ value using GD optimization to maximize \eqref{eq:soft_obj}. 
In METHODS, we validate
this by linking BP update rules with GD. Given the physical error rate $p$, we can determine $\alpha$ from a locally linear function of $p$ using the continuity of the reward function \eqref{eq:soft_obj} (see \eqref{eq:llr_tanh_a}--\eqref{eq:a_vs_g}).
Furthermore, we introduce an optional step \eqref{eq:MBP2_GD} to adjust the initial LLR $\Lambda_j$ to have enough magnitude $|\Lambda_\text{GD}|$, ensuring enough update gradients in harmful stopping sets.  
This is performed whenever the updated LLR $\Gamma_j$ for index $j$ does not have enough magnitude  (i.e., $|\Gamma_j|<|\Lambda_\text{GD}|$) after every $T_\text{GD}$ iterations.
%

In BP, the message $\Delta_{i\to j}$ is typically computed by
    $
    \Delta_{i\to j} = (-1)^{\vs_i}\boxplus_{j'\in\sN(i)\setminus\{j\}} \Gamma_{j'\to i}.
    $
For efficiency, we compute it by 
    \begin{equation} \label{eq:bminus} 
    \Delta_{i\to j} \,= \,\Delta_i \boxminus \Gamma_{j\to i} \,:=\, 2\tanh^{-1}\left( \tanh\tfrac{\Delta_i}{2} \,\,/\, \tanh\tfrac{\Gamma_{j\to i}}{2} \right),
    \end{equation}
where 
    $
    \Delta_{i} = (-1)^{\vs_i}\boxplus_{j\in\sN(i)} \Gamma_{j\to i}
    $
is computed first 
and every $\Gamma_{j\to i}\ne 0$. 
This ensures $0 < |\Delta_i| < |\Gamma_{j\to i}|$ for numerical stability, which is guaranteed by the function $\text{soft}(\cdot)$ in \eqref{eq:soft}.
Compared to the typical computation, this reduces the BP complexity from $O(2n j_2 (j_2 - 1))$ to $O(2n j_2)$ per-iteration, where $j_2$ is the mean column-weight of $\vH$.

Additionally, we observe that altering the message-update schedule can significantly improve decoding convergence for highly degenerate codes 
(Supplementary Note~8).

\begin{figure}          
\begin{algorithm}[H]    
\begin{flushleft}
\caption{: AMBP$_2$} \label{alg:AMBP2}
	{\bf Input}:
	The same as in Algorithm~\ref{alg:MBP2} but further with \\
	\quad a sequence of real values $\alpha_1 > \alpha_2 >\dots > \alpha_\ell > 0$. \\[2pt]

	{\bf Initialization}: Let $i=1$. 

	{\bf MBP Step}: Run MBP$_2(\vH,\, \vs,\, T_{\max},\, \alpha_i,\, \{\Lambda_j\})$,
	\begin{itemize}
		\item[] which returns an indicator ``CONVERGE'' or ``FAIL'' with an estimated $\hat{\vE} \in \{0,1\}^{2n}$.
	\end{itemize}

	{\bf Adaptive Check}: 
	\begin{itemize}
		\item If the return indicator is ``CONVERGE'', return ``CONVERGE'' (with $\alpha^* = \alpha_i$ and the estimated $\hat{\vE}$);
		\item Otherwise, let $i\leftarrow i+1$; if $i>\ell$, return ``FAIL''; 
		\item Otherwise, repeat from the MBP Step.
	\end{itemize}
\end{flushleft}
\end{algorithm}
~\\[-20pt]
{\footnotesize For erasure decoding, all MBP instances with different $\alpha$ values can run in parallel, halting all instances as soon as one finds a feasible solution.}
\end{figure}

As discussed earlier, we can select the $\alpha$ value from GD, but it is based on the reward function \eqref{eq:soft_obj} and not necessarily optimal. 
First, the function in \eqref{eq:soft_obj} is not concave and has vanishing gradients when encountering stopping sets, so a gradient-based step-size search does not ensure convergence.
Second, directly using a large step-size may lead to oscillations and prevent the algorithm from converging. To mitigate this, we can test a list of decoders with step-sizes ranging from small to large step-sizes ($1/\alpha$).
Thus, we also propose an adaptive MBP$_2$ (AMBP$_2$), as in Algorithm~\ref{alg:AMBP2}, which tests a list of $\alpha_1>\dots>\alpha_\ell$ and adaptively select a proper $\alpha^*$.

\subsection*{MBP$_4$, and unified (A)MBP$_q$ notation, $q=2$ or $4$}

We also consider quaternary BP (BP$_4$). 
This BP leverages the correlation between bits $j$ and $j+n$ via the Pauli check matrix $H \in\{I, X, Y, Z\}^{m \times n}$ for error correction, enhancing accuracy in both classical \cite{DM98b} and quantum settings \cite{KL20}, when compared to binary BP.
Indirectly, this can be interpreted as a turbo update between the two error components $(\vE^X | \vE^Z)$ \cite{MMC98}, reducing the likelihood of the decoder being halted by binary stopping sets in either the $X$ or $Z$ portion.
Preliminary simulations show very good accuracy, so we do not design a GD step for quaternary BP.

For quaternary BP in the quantum setting, we express the variable-to-check $(j\to i)$ message, following \cite{KL20,KL21}, as 
    \begin{align}
    &\quad \ln\medmath{\frac{\Pr(\text{Pauli error $E_j$ and entry $H_{ij}$ commute})}{\Pr(\text{Pauli error $E_j$ and entry $H_{ij}$ anti-commute})}},\\
    &\text{e.g., }
    = \ln\tfrac{p_j^I+p_j^X}{p_j^Y+p_j^Z}
    \text{~~for initial LLR $(j\to i)$ if $H_{ij}=X$.}
    \end{align}
    %
\ed{Extending the derivations in \eqref{eq:llr_tanh_a}--\eqref{eq:a_vs_g}}, we obtain quaternary MBP (MBP$_4$) and its adaptive version (AMBP$_4$) \cite{KL22}.

\begin{figure}
\begin{algorithm}[H]
\begin{flushleft}
\caption{: Erasure decoding with (A)MBP$_q$, $q=2$ or 4} \label{alg:BP_era}
	{\bf Input}:\\ 
	\quad A check matrix $H\in\{I,X,Y,Z\}^{m\times n}$, a~syndrome $\vs\in\{0,1\}^m$,\\ 
	\quad erased coordinates $\br\subseteq\{1,2,\dots,n\}$, 
	an~integer $T_{\max}>0$,\\
	\quad and a set of real $\alpha_1>\alpha_2>\dots>\alpha_\ell>0$. 
    \\[3pt]
    \quad Include $T_\text{GD}$ and $|\Lambda_\text{GD}|$ if the GD step is applied to (A)MBP$_2$.
    \\[3pt]

	{\bf Initialization.} 
	\begin{itemize}
	\item Set $(p^I_j,p^X_j,p^Y_j,p^Z_j) = (\frac{1}{4},\frac{1}{4},\frac{1}{4},\frac{1}{4})$ for $j\in\br$\\  
	and $(p^I_j,p^X_j,p^Y_j,p^Z_j) = (1,0,0,0)$ for $j\in\bar{\br}$.\\
    (Convert to binary distributions \eqref{eq:main_p2_init} if $q=2$.)
	\end{itemize}
		 
	{\bf Decoding.} 
	\begin{itemize}
	\item Run (A)MBP$_q(H,\, \vs,\, T_{\max},\, \{p^I_j,p^X_j,p^Y_j,p^Z_j\}_{j=1}^{n},\, \{\alpha_i\}_{i=1}^\ell)$,\\ 
    with optional GD step with input $(T_\text{GD}, |\Lambda_{GD}|)$ if applied, and obtain an estimate $\hat{E}\in\{I,X,Y,Z\}^n$.
	\end{itemize}

	{\bf Output.} 
	\begin{itemize}
	\item Return $\hat{\vE}\in\{0,1\}^{2n}$ 
        and $\hat \vs = \langle\hat{\vE},\vH\rangle$, where $\hat{\vE}$ and $\vH$ are the binary version of $\hat{E}$ and $H$, respectively.
	\end{itemize}
\end{flushleft}
\end{algorithm}
\end{figure}

Given a Pauli check matrix $H\in\{I,X,Y,Z\}^{m\times n}$, which uniquely determines a binary check matrix $\vH\in\{0,1\}^{m\times 2n}$,
we consider a unified (A)MBP notation: let $\ell = 1$ for MBP and $\ell > 1$ for AMBP, and refer the decoder by
    $$
    \text{(A)MBP}_q(H,\, \vs,\, T_{\max},\, \{p^I_j,p^X_j,p^Y_j,p^Z_j\}_{j=1}^{n},\, \{\alpha_i\}_{i=1}^\ell),
    $$
with $q=4$ for (A)MBP$_4$ and $q=2$ for (A)MBP$_2$.
We state a unified (A)MBP$_q$ algorithm in Algorithm~\ref{alg:BP_era}.

For (A)MBP$_2$, if the optional GD step is applied, we refer the algorithm as GD (A)MBP$_2$ in simulations.

\section*{Discussions}  \label{sec:con}

We analysed the general erasure decoding problem, and proposed efficient decoders suitable for this task. 
Topological codes were known to achieve capacity with vanishing code rate \cite{SBD09,Ohz12,DZ20}. 
We construct various rate bicycle and LP codes, and their MLD thresholds attain the erasure capacity (Fig.~\ref{fig:bnd_logx}).
Soft (A)MBP decoders behave like MLD at low logical error rate, achieving near-capacity performance with a small gap.
The GD flipping decoder performs well on bicycle and LP codes, despite its extremely low complexity.

The Gaussian decoder resembles the core idea of ordered statistics decoding (OSD) \cite{FL95}, solving linear equations restricted to unreliable error estimates. Combining BP with OSD improves BP accuracy for decoding depolarizing errors \cite{PK21,RWBC20}. 
For erasures, a similar \mbox{BP + Gaussian} approach can achieve MLD performance. In this setting, Gaussian elimination targets only erasures left by BP. 
This reduces complexity compared to using the Gaussian decoder alone, though the overall complexity remains $O(n^3)$; 
see Supplementary Note~10 
for details.

Our results extend beyond erasure errors. 
First, we consider deletion errors, which follow the same model as erasure errors but without knowledge of the error locations.
We apply our results on concatenated QLDPC codes with permutation-invariant (PI) codes \cite{Rus00,PoR04,Yin14,ouyang2015permutation,OUYANG201743,Yin21,AAB24} as inner codes. 
Deletion errors can be catastrophic for stabilizer codes, but PI codes can correct them \cite{Yin21,shibayama2021permutation}.
Unlike erasures, a primary challenge of deletion errors is the lack of error location information, making decoding more difficult, as misidentifying even one location can cause failure.
Accurately locating all deletions to obtain erasures aids decoding but is time-consuming, increasing the physical error rate. 
Thus, accepting imperfect erasure conversions, where some leakage errors remain as localized deletions, can be more effective. 
%
We assume a local deletion model where each block of $\ell$ physical qubits has an identifiable boundary, protected by an $[[\ell,1]]$ PI code, converting deletions into erasures.
Figure~\ref{fig:Convert} shows the deletion-to-erasure conversion performance in a channel with deletion probability~$\epsilon$, using a PI code that corrects $t$ deletions (see legend). 
The $[[4,1]]$ PI code \cite{Yin14,nakayama2020first,HagiwaraISIT2020} converts local deletions with probability $\epsilon=0.385$ into erasure errors of rate $p<0.5$, which topological codes can handle. 
%
%
At $\epsilon=0.1$, an $[[\ell,1]]$ code with $\ell<15$ achieves a converted $p<0.1$, allowing a QLDPC code rate above 3/4 (Fig.~\ref{fig:bnd_logx}).
In particular, if the per-qubit deletion probability is 0.01, then using a 7-qubit PI code can bring the logical erasure rate to below $10^{-4}$, which allows QLDPC codes to attain a rate of essentially 1, which is close to the maximum possible rate.

\begin{figure}
\centering\includegraphics[width=0.48\textwidth]{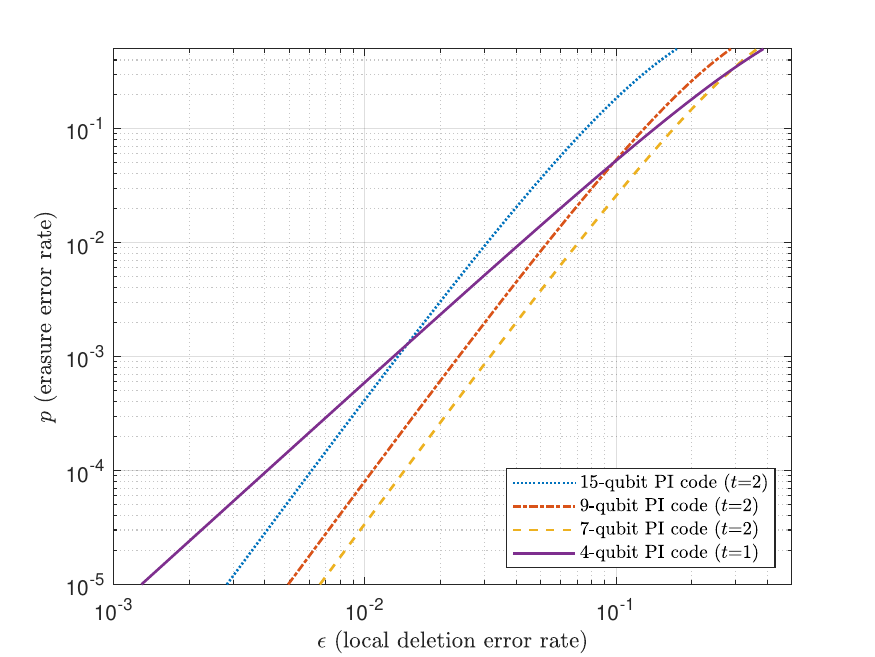}
    \caption{
    Local deletion to erasure rate conversion using $[[\ell,1]]$ PI codes for $\ell = 4,7,9,15$ 
    (\cite{Yin14,HagiwaraISIT2020}, 
    \cite{PoR04}, 
    \cite{Rus00}, 
    \cite{AAB24}).
    The $[[4,1]]$ PI code converts $\epsilon=0.385$ to $p=0.5$, which can then be handled by topological codes.
    At $\epsilon=0.1$, an $[[\ell<15,\,1]]$ PI code achieves $p<0.1$, allowing a QLDPC code rate above 3/4 (Fig.~\ref{fig:bnd_logx}).
    At $\epsilon=0.01$, using a 7-qubit PI code achieves $p<10^{-4}$, which allows QLDPC codes to attain a rate of essentially 1.
    } \label{fig:Convert} 
\end{figure}

Second, in a large system, after converting dominant noise into erasures at known locations, we model additional errors that occur as qubit depolarizing errors. 
This mixed error model is discussed in \cite{DN21} using peeling on 2D codes.
In contrast, (A)MBP can decode general QLDPC codes in this error model. 
We show this using a $[[1054,140]]$ LP code \ed{(from \cite{PK22})} 
with the \ed{simulation setup in METHODS}.
With (A)MBP$_4$, we adapt by adjusting the initial error distribution for the non-erased qubits to
    $(p^I,p^X,p^Y,p^Z) = 
    { (1-p_0,\,\frac{p_0}{3},\,\frac{p_0}{3},\,\frac{p_0}{3}) },$
where $p_0$ can be set to either the depolarizing error rate $p_\text{dep}$ or a suitable fixed value $<p_\text{dep}$ \cite{HFI12,KL22}. 
With (A)MBP$_2$, the initial error distribution is $(p^{(0)},p^{(1)})$, where
$p^{(1)} = p^X+p^Y=p^Z+p^Y$ 
and $p^{(0)}=1-p^{(1)}.$
The adjustment enables soft BP decoders to simultaneously correct both erasure and depolarizing errors.
Figure~\ref{fig:Eras_Dep} shows our decoder accuracy, assuming $p_\text{dep}=0.01$. 
GD Flip-BP$_2$ cannot handle depolarizing errors, so we use MBP$_2$ with $\alpha=1$ as a soft extension. 
All decoders successfully handle mixed errors. 
Among them, AMBP$_4$ shows the best accuracy, as depolarizing errors exhibit strong correlations between $X$ and $Z$ errors.

Our findings highlight the proposed decoders' adaptability and accuracy in handling various error types, including erasures, deletions, and depolarizing errors.
%
%
We expect the versatility and high performance of our proposed decoders to contribute to the development of quantum LDPC codes and fault-tolerant quantum computation.

For practical use, the $\tanh$ rule~\eqref{eq:bsum} can be approximated by the min-sum rule, which selects the minimum magnitude among the input LLRs and combines signs separately \cite{CF02b,CDE+05,VGRV21}. It retains soft information and may outperform flip BP, motivating its study in quantum erasure correction.

\ky{Quantum channel limits have long been known \cite{BDSW96,GotPhD}, \cite{BDS97}. 
For the depolarizing channel, the hashing bound (also known as the quantum Hamming bound) sets the capacity for non-degenerate codes, though degenerate codes may slightly exceed it at low physical error rates \cite{BDSW96,GotPhD}.
However, near-hashing-bound decoding has emerged only recently \cite{KK25}. 
By contrast, although the capacity of the erasure channel is known \cite{BDS97}, comparable decoding results for erasure noise models remain largely unexplored.
Our work closes this gap by providing decoders simultaneously achieving near–erasure-capacity accuracy and linear-time complexity.}

\ky{A recent review \cite{Dem+24} surveys decoding algorithms for surface codes, including erasure decoding, although focusing on the union-find (peeling) decoder.
We recall that peeling decoder + clusters decomposition with  unconstrained size ($O(n^3)$ complexity) achieves MLD on GHP codes \cite{YGP24}. 
More recently, a quantum Maxwell erasure decoder for CSS QLDPC codes \cite{Fre+26} extends peeling with guessing, which interpolates between linear-time decoding (constant guessing budget) and MLD (unconstrained budget). Its asymptotic performance has been analyzed and simulated for bivariate bicycle and quantum Tanner codes.}

\ky{Circuit-level decoding plays a key role in fault-tolerant quantum computation \cite{DKLP02}, with substantial progress focused mainly on depolarizing and gate noise with unknown locations, notably for topological codes \cite{WHP03,RHG07,WFSH10,FWH12} and bivariate bicycle codes \cite{Bra+24}. Implementing circuit-level simulations can be assisted by the stabilizer-circuit tool {\tt Stim} \cite{Gid21s}, following standard practice in circuit-level decoding under unknown error locations \cite{Sko+23,Anqi+24,Hil+25,Mul+25} or erasures \cite{GRK25,Gu+25,Per+25,Bar+26}.}

\ky{Our approach can be generalized similarly using {\tt Stim}. Alternatively, to retain the benefits of non-binary (A)MBP decoding, our approach extends to phenomenological or circuit-level noise by representing the syndrome, data, and two-qubit gate errors with mixed-alphabet variables in BP \cite{KL25,KL24a}. 
While circuit-level errors with unknown locations are more challenging, prior results demonstrate the viability of this mixed-alphabet approach \cite{KL25,KL24a}.} 

\ky{A detailed treatment of circuit-level erasures, which requires additional erasure-conversion circuits \cite{GRK25,Gu+25,Per+25,Bar+26}, lies beyond the scope of this paper.}

\begin{figure}
\centering\includegraphics[width=0.48\textwidth]{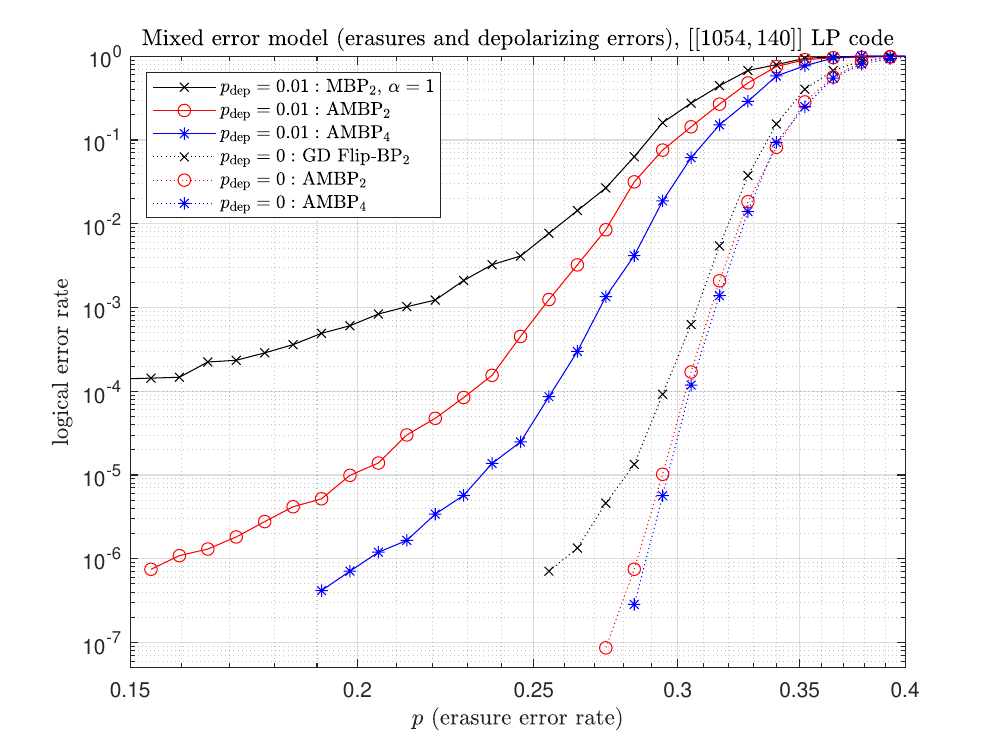}
	\caption{
	BP accuracy for mixed erasure and depolarizing errors using the $[[1054,140]]$ LP code. The curves with only erasure errors are provided for reference ($p_\text{dep}=0$). AMBP$_4$ maintains strong accuracy in the mixed error model ($p_\text{dep}=0.01$).
	} \label{fig:Eras_Dep} 
\end{figure}

\section*{Methods} \label{sec:methods}

\subsection*{BP update rules and the relation with GD} \label{sec:bp2_rules}

For BP$_2$ syndrome decoding in the erasure channel, we have to satisfy two constraints, $\br$ and $\vs$ (i.e., the erasure locations and the syndrome).
The erasure information $\br$ is embedded in the initial distributions \eqref{eq:main_p2_init}. 
The syndrome information $\vs$ is included in BP$_2$ to update the initial LLR
    $\Lambda_j = \ln\frac{\Pr(\vE_j=0)}{\Pr(\vE_j=1)}$ 
to 
    $\Gamma_j = \ln\medmath{\frac
    {\Pr(\vE_j=0~|~\text{syndrome is } \vs)}
    {\Pr(\vE_j=1~|~\text{syndrome is } \vs)}}$
for each binary error $\vE_j$ for $j=1,2,\dots,2n$.
It is well-known that normalized BP$_2$ \cite{CF02a,CDE+05} provides an efficient computation 
    \begin{align}  
    &\textstyle
    \Gamma_j = \Lambda_j + \frac{1}{\alpha} \, \sum\limits_{i\in\sM(j)} (-1)^{\vs_i} 2\tanh^{-1}\left( \prod\limits_{j'\in\sN(i) \setminus \{j\}} \tanh\tfrac{\Gamma_{j'\to i}}{2} \right) \label{eq:llr_tanh_a}
    \\
    &\textstyle
    ~~~ = \Lambda_j + \frac{1}{\alpha} \, \sum\limits_{i\in\sM(j)} \Delta_{i\to j} \label{eq:llr_by_Delta}
    \end{align}
like Algorithm~\ref{alg:MBP2}. 
We justify this computation for the first iteration, with $\alpha=1$, in 
Supplementary Note~5 (part~a).

To determine the value of $\alpha$ in Algorithm~\ref{alg:MBP2}, we relate BP to GD \cite{LBB98,KL22}.
Define the satisfaction at check node~$i$ as 
    $
    (-1)^{\vs_i} 2\tanh^{-1} \left( \prod_{j\in\sN(i)} \tanh\frac{\Gamma_{j}}{2} \right).
    $
Summing this for all check nodes, we get the reward function \eqref{eq:soft_obj}, which is:
    \begin{equation} \label{eq:J_S} 
    \textstyle
    S(\Gamma_1,\dots,\Gamma_{2n}) = \sum\limits_{i=1}^m (-1)^{\vs_i} 2\tanh^{-1} \Bigg( \prod\limits_{j\in\sN(i)} \tanh\frac{\Gamma_{j}}{2} \Bigg).
    \end{equation}
Then 
    \begin{equation} \label{eq:par_der} 
    \frac{\partial S}{\partial\Gamma_j} \textstyle
    = \sum\limits_{i\in\sM(j)} g_{ij} ~ (-1)^{\vs_i}\prod\limits_{j'\in\sN(i) \setminus \{j\}} \tanh\frac{\Gamma_{j'}}{2},
    \end{equation}
where $g_{ij}$ is a positive term, as a function of $(\Gamma_1,\dots,\Gamma_{2n})$, as
    \begin{equation} \label{eq:g_ij}
    g_{ij}(\Gamma_1,\dots,\Gamma_{2n}) 
    = \frac
    { 1-\tanh^2\frac{\Gamma_j}{2} }
    { 1-\left( \prod_{j'\in\sN(i)} \tanh\frac{\Gamma_{j'}}{2} \right)^2 }
    > 0.
    \end{equation}
Notice that $|\Gamma_{j\to i}|<\infty$ due to the use of $\text{soft}(\cdot)$ in \eqref{eq:soft}.
For non-erased part, $\Lambda_j = +\infty$, enforcing $\Gamma_j= \Lambda_j$ in \eqref{eq:llr_tanh_a} and fixing bit $j$ decision as 0.
For erased part, $\Lambda_j=0$ in \eqref{eq:llr_tanh_a}, so the other term determines the decision of bit $j$. 
Comparing this term in \eqref{eq:llr_tanh_a} and \eqref{eq:par_der}, and noting that $\tanh^{-1}(\cdot)$ and $g_{ij}>0$ (as a multiplier) preserve the sign, we conclude that for BP, we should choose an $\alpha$ value such that
    \begin{equation} \label{eq:a_vs_g}
    \frac{1}{\alpha}\, \textstyle 
    2\tanh^{-1}\left( \prod_{j'\in\sN(i)\atop j'\ne j} \tanh\tfrac{\Gamma_{j'\to i}}{2} \right)
    \propto 
    g_{ij} \prod_{j'\in\sN(i) \atop j'\ne j} \tanh\frac{\Gamma_{j'}}{2}
    \end{equation}
to minimize \eqref{eq:J_S}. 
Using \eqref{eq:a_vs_g}, we can express $\alpha$ as a locally linear function of $p$. 
A proper initial $\alpha$ value from $p$ notably reduces AMBP iteration count 
\ed{(see Table~\ref{tb:iter_GHP}).}

We determine $\alpha$ by associating $\Gamma_j$ with the physical error rate $p$. This step is crucial to prevent wrong decisions in the first iteration that would propagate incorrect information in subsequent iterations. 
As all functions in \eqref{eq:a_vs_g} are continuous, $\alpha$ is locally linear in $\Gamma_j$. Associating $\Gamma_j$ with the channel error rate $p$, we express $\alpha$ as a locally linear function of $p$.
We also associate $\alpha$ with the GD parameters in Algorithm~\ref{alg:MBP2} and obtain $T_\text{GD}=5$ and $|\Lambda_\text{GD}|=0.25$ for enough gradients for algorithm updates. 
See Supplementary Note~5 (part~c).



\subsection*{Simulation setup} \label{sec:main_sim}

We conduct computer simulations for HP \cite{TZ09_14} and GHP codes \cite{PK21}, bicycle codes \cite{MMM04}, LP codes \cite{PK22}, and topological toric \cite{Kit97_03,DKLP02,BM07} and XZZX codes \cite{KDP11,THD12,ATBFB21}.
%
Supplementary Note~7 provides more information about the codes.
We test BP and MLD on theses codes.

For our BP decoders, we set $T_{\max}=100$. The average iteration count scales as $O(\log\log n)$ at relevant error rates, as discussed later.
%
For (A)MBP, unless otherwise stated, we use the group random schedule (Supplementary Note~8), 
and select $\alpha^*$ from $1.2,1.19,\dots,0.3$ is using adaptive scheme.

\subsection*{A sphere-packing bound for accuracy evaluation}

We aim for correcting $t$ erasures with $t\propto n$. 
If a qubit is erased, it encounters $I,X,Y, \text{ or } Z$ with equal probability $1/4$. 
Using \eqref{eq:like_dep_ch}, 
a qubit encounters~$I$ with probability $1-p+\frac{p}{4} = 1-\frac{3}{4}p$, and encounters $X,Y, \text{ or } Z$ with equal probability $\frac{p}{4}$. 
Thus, before specifying $\br$, we have 
	\begin{equation} \label{eq:uc_prob}
	(p^I,p^X,p^Y,p^Z) = \textstyle (1-\frac{3p}{4},\, \frac{p}{4},\, \frac{p}{4},\, \frac{p}{4}).
	\end{equation} 
When $t\propto n$ are both large, we expect most occurred errors $E\in\{I,X,Y,Z\}^n$ to have approximately $\frac{3}{4}t$ components as $X,Y,$ or $Z$. 
The probability that the number of such components deviates from $\frac{3}{4}t$ by $\delta t$ for any $\delta>0$ is exponentially suppressed in $t\propto n$.
\begin{definition} \label{def:eBDD}
The logical error rate of \emph{erasure bounded-distance-decoding} (eBDD) with parameters $n,t,$ and $p$ is  
    \begin{equation} \label{eq:eBDD}
    P_\text{eBDD}(n,t,p) := \textstyle \sum\limits_{j=\lceil\frac{3}{4}t\rceil+1}^n {n\choose j} (\tfrac{3}{4}p)^{j}(1-\tfrac{3}{4}p)^{n-j}.
    \end{equation}    
\end{definition}

If a code family has decoding accuracy closely matching $P_\text{eBDD}(n,t,p)$ with a constant $\frac{t}{n}=p_\text{th}$ as $n$ grows, then $p_\text{th}$ serves an estimated threshold for the code family. 
With a fixed $\frac{t}{n}=p_\text{th}$ in \eqref{eq:eBDD}, at any physical error rate $p < p_\text{th}$, the function $P_\text{eBDD}(n,t,p)$ decays exponentially as $n$ increases (see Supplementary Note~6).
The accuracy estimated from eBDD is more reliable than the intersection of performance curves, particularly at low error rates (see Fig.~\ref{fig:decs_on_LP118}).

\subsection*{Performance of capacity-achieving codes} \label{sec:cap_codes}

\begin{figure}
    \hspace*{-6mm}
    \centering\includegraphics[width=0.56\textwidth]{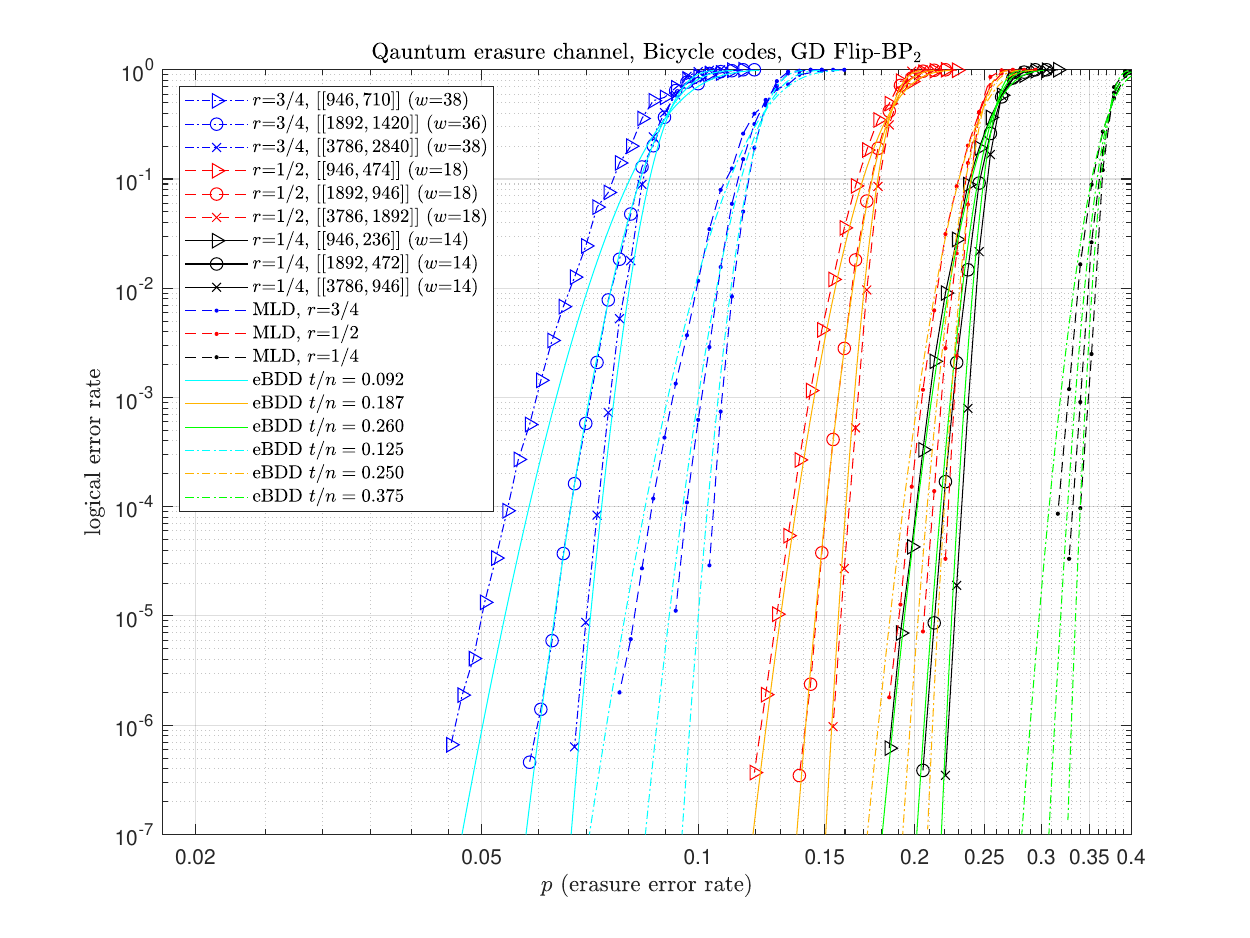}
    \caption{
    Accuracy of GD Flip-BP$_2$ and MLD on bicycle codes. 
	Decoders (A)MBP$_q$ have accuracy comparable to GD Flip-BP$_2$ on these codes. 
    BP threshold for rate 1/4 can be improved from 0.26 to 0.285 by using row-weight $w=12$ 
	\ed{(as shown in Supplementary Note~7 (part~c))}.
	} \label{fig:bic_wt_up}
\end{figure}

We consider moderate-to-high-rate bicycle codes, low-rate LP codes, and vanishing-rate 2D topological codes.
For bicycle and LP code families,
we estimate the threshold values using eBDD, by comparing the decoder accuracy with \eqref{eq:eBDD} given fixed $t/n$.

First, 
we construct bicycle codes with rates 1/4, 1/2, and 3/4. 
    A~key step in bicycle construction is row-deletion, which requires heuristics \cite{MMM04},
    discussed in \ed{Supplementary Note~7 (part~c)}.
GD Flip-BP$_2$ suffices for decoding bicycle codes, with accuracy shown in Fig.~\ref{fig:bic_wt_up}, including MLD for comparison.
Soft decoders (A)MBP have accuracy comparable to GD~Flip-BP$_2$ on these codes.
Using eBDD, we estimate the threshold values, as shown in the legend of Fig.~\ref{fig:bic_wt_up}. 
Bicycle codes exhibit capacity accuracy ($r=1-2p$) with MLD threshold values at $p=0.125, 0.25, 0.375$ for rates $r=3/4, 1/2, 1/4$, respectively. 
The BP threshold follows closely, with a small gap at $p=0.092, 0.187, 0.26$ for these rates.
BP threshold for rate 1/4 can be improved from 0.26 to 0.285 by using row-weight $w=12$, by tolerating some error floor 
(see \ed{Supplementary Note~7 (part~c)}).

\begin{figure}
    \centering
    \includegraphics[width=0.5\textwidth]{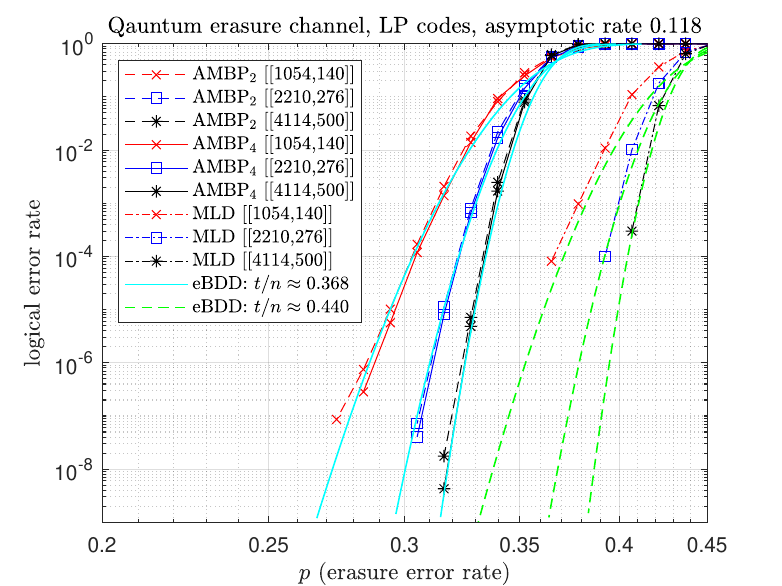}
    \caption{
    Accuracy of AMBP and MLD on rate 0.118 LP codes. 
    Our AMBP$_q$ exhibit a threshold of $p_\text{th} = 0.368$. 
    MLD curves suggest a higher threshold of $p_\text{th} = 0.44$.
	\ed{(Results for rate 0.04 LP codes are in Supplementary Note~7 (part~d).)}
    } \label{fig:LP118}
\end{figure}

\begin{figure}
    \centering \includegraphics[width=0.5\textwidth]{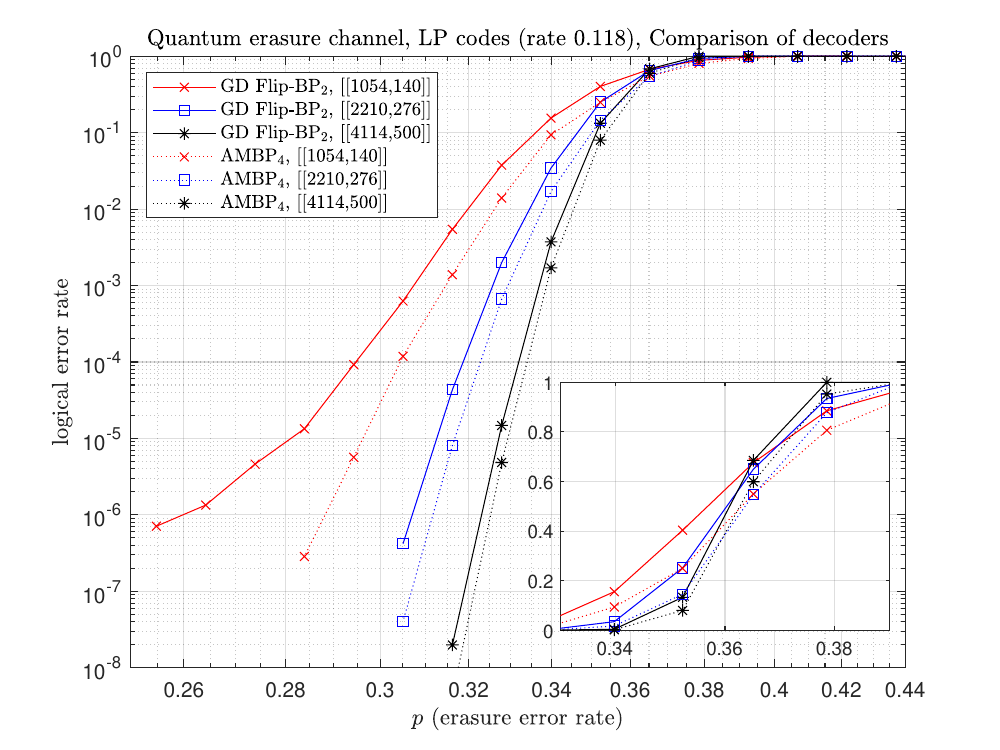}
	\caption{
	GD Flip-BP$_2$ vs AMBP$_4$ on rate 0.118 LP codes.
    AMBP$_4$ has a threshold of 0.368 (Fig.~\ref{fig:LP118}). 
    GD Flip-BP$_2$ closely matches AMBP$_4$ as $n$ increases. Also note that good curve intersection does not guarantee accuracy at low error rates (red line). 
	} \label{fig:decs_on_LP118}
\end{figure}

Second, we construct quantum LP codes \cite[Example~3]{PK22} from classical quasi-cyclic (QC) matrices \cite{Fos04,Tan+04,Mal07}.
We construct LP code families with rates of 0.118 and 0.04. 
For rate 0.118, Figure~\ref{fig:LP118} show the accuracy of AMBP and MLD, along with eBDD curves.
%
LP codes exhibit strong accuracy with MLD, achieving a threshold of $p^\text{(MLD)}_\text{th}=0.44$ at rate $r=0.118$, closely matching the capacity $r = 1 - 2p$.
AMBP decoders also achieve a large threshold of $p^\text{(BP)}_\text{th} = 0.368$. 
For this code family, GD Flip-BP$_2$ performs almost as well as AMBP$_4$ when the code length increases, as shown in Fig.~\ref{fig:decs_on_LP118}. This figure also indicate that good curve intersection does not guarantee accuracy at low error rates; see the red line.

\ed{Estimated thresholds for rate 0.04 LP codes are $p^\text{(MLD)}_\text{th}=0.48$ and $p^\text{(BP)}_\text{th}=0.405$ (Supplementary Note~7 (part~d)).}

GD Flip-BP$_2$ has very low implementation complexity, though its accuracy may depend on the code construction.
Figures~\ref{fig:bic_wt_up} and~\ref{fig:decs_on_LP118} show that it can achieve high accuracy with bicycle and LP constructions.

\begin{figure}
	\centering
	\includegraphics[width=0.5\textwidth]{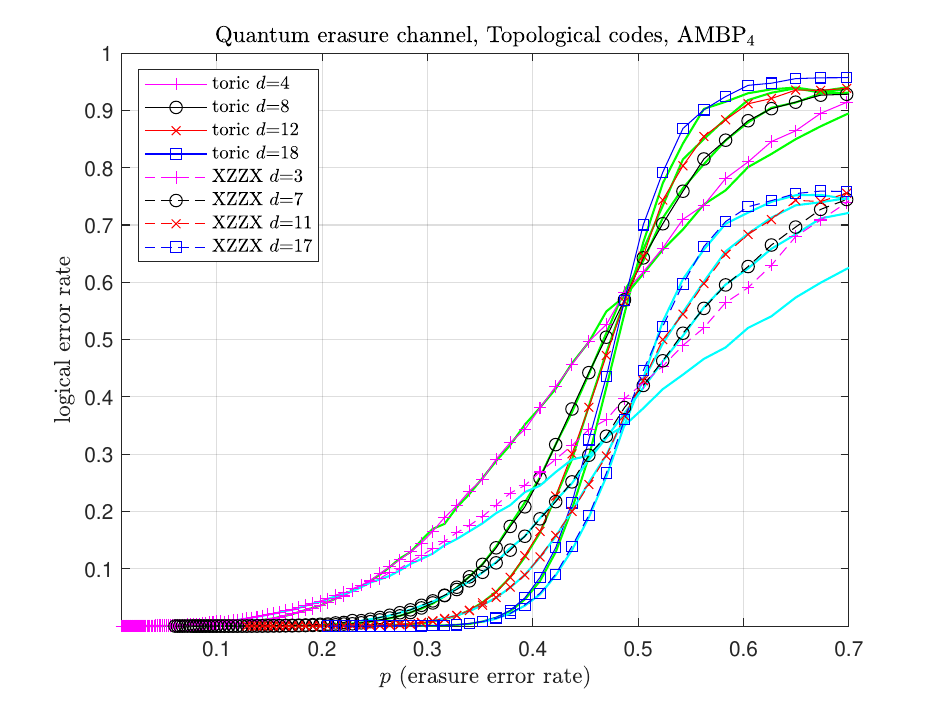}
	\caption{
        Accuracy of AMBP$_4$ and MLD on toric and XZZX codes. Both decoders achieve a threshold of 0.5. 
        MLD curves are shown in green and cyan for the toric and XZZX codes, respectively.
        Both decoders have the same accuracy at low error rates 
		\ed{(as a log-scale figure in Supplementary Note~7 (part~e))}.
	} \label{fig:2D_lin} 
\end{figure}

Third, for 2D topological codes, we consider (rotated) toric codes ${[[L^2,\,2,\,L]]}$ for even $L$ \cite{Kit97_03,BM07} and (twisted) XZZX codes ${[[(d^2+1)/2,\,1,\,d]]}$ for odd $d$ \cite{KDP11} (see a summary of code structures in \cite[Table~I and figures]{KL22isit}). 
For 2D codes, 
the intersection of performance curves, known as the Nishimori point, provides an estimate of the critical error probability (threshold) for MLD \cite{DKLP02}. 
Among our decoders, AMBP$_4$ closely matches MLD accuracy for these codes, as shown in Fig.~\ref{fig:2D_lin}. Both decoders have a threshold of 0.5,
consistent with \cite{SBD09,Ohz12,DZ20}. 
For AMBP, we use a smaller $\alpha_1=0.95$ to increase the step-size, handling mores erasures while maintaining decoder accuracy.
To examine accuracy at low error rates, 
we replot Fig.~\ref{fig:2D_lin} on a log scale, showing that AMBP$_4$ matches MLD 
\ed{on these 2D codes (see Supplementary Note~7 (part~e))}.

\subsection*{Selection of the $\alpha$ value} \label{sec:opt_a}

We demonstrate how to select a proper $\alpha$ value for MBP or an appropriate initial value $\alpha_1$ for AMBP. 
A code with a transition point in the performance curve is convenient for this analysis.
Consider the $[[882,48,16]]$ GHP code from \cite{PK21}. The accuracy of MBP$_2$ with different $\alpha$ values is shown in Fig.~\ref{fig:GHP_a}. 
Three transition points occur around $p=0.32, 0.33,0.34$, for $\alpha=1.2, 1.0, 0.8$, respectively. 
This locally linear trend aligns with our expectation 
\ed{(see Supplementary Note~5).}
We establish a linear relation by $(p,\alpha) = (0.32, 1.2)$ and $(0.34, 0.8 + \text{margin})$, where the margin is 0.1 to prevent $\alpha$ from being too small. The obtained relation is: ${\alpha = -15p + 6}$.
Hence, we propose a function: 
    \begin{equation} \label{eq:a_func_p} 
    \alpha = \text{func}(p) := \max(\min(-15p + 6,\, 1.2),~ 0.3), 
    \end{equation} 
where 1.2 and 0.3 are the upper and lower bounds, respectively. 
Additional simulations show that a wider range of $\alpha$ values does not lead to further improvement. 
With this $\alpha$ selection strategy \eqref{eq:a_func_p}, MBP achieves accuracy as the lower envelope of the curves in Fig.~\ref{fig:GHP_a}.

\begin{figure}
    \centering
    \includegraphics[width=0.5\textwidth]{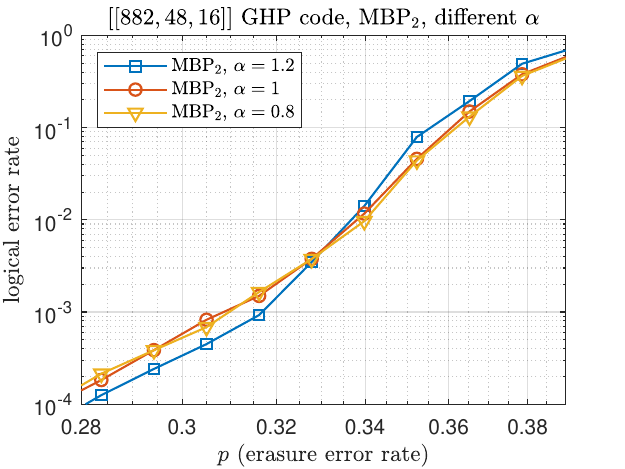}
    \caption{
    Accuracy of MBP$_2$ on the $[[882,48,16]]$ GHP code. Three transition points occur around $p=0.32, 0.33, 0.34$, for $\alpha=1.2, 1.0, 0.8$, respectively, establishing a linear relation.
    } \label{fig:GHP_a}
\end{figure}

\begin{figure}
    \centering \includegraphics[width=0.5\textwidth]{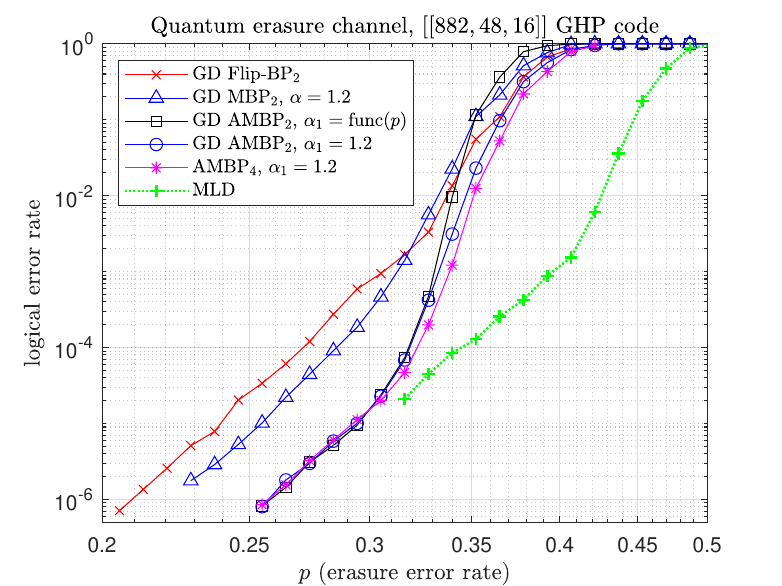}
    \caption{
    Accuracy of several decoders on the $[[882,48,16]]$ GHP code. 
    For BP decoders, AMBP-type decoders achieve the best accuracy, and can maintain a low iteration count when using $\alpha_1 = \text{func}(p)$ as in~\eqref{eq:a1_func}. The runtime is compared in Table~\ref{tb:iter_GHP}. 
	\ed{More HP and GHP results are in Supplementary Note~7 (part~a).}
    } \label{fig:GHP}
\end{figure}

For AMBP-type decoders, we evaluate two strategies:
    \begin{align}
    & \alpha^*\in\{1.2,~ 1.19,~ \dots,~ 0.3\}, \quad \text{ or } \label{eq:a1_fix} \\
    & \alpha^*\in\{\text{func}(p),~ \text{func}(p)-0.01,~ \dots,~ 0.3\}. \label{eq:a1_func}
    \end{align}
We present the decoder accuracy in Fig.~\ref{fig:GHP}. 
As shown, the two strategies exhibit similar accuracy. 
Both strategies show similar average iteration counts when $p\le 0.328$ (logical error rate $<10^{-4}$), as detailed in Table~\ref{tb:iter_GHP}.
A significant difference arises at higher error rates, such as $p=0.392$. 
The strategy in \eqref{eq:a1_func} selects a proper $\alpha_1$ for a given physical error rate $p$, making AMBP as efficient as MBP. 
An illustration of the two strategies is shown in \ed{Supplementary Note~5 (see S~Fig.~5 therein)}.

\medskip

\begin{table}
\caption{
Average BP iteration counts for the $[[882,48,16]]$ GHP code. 
We include more decoders for reference than in Fig.~\ref{fig:GHP}.
All BPs remain efficient at low error rates. 
For AMBP, two strategies \eqref{eq:a1_fix}~(fixed $\alpha_1$) and \eqref{eq:a1_func}~($\alpha_1 = \text{func}(p)$) are compared, where \eqref{eq:a1_func} provides much faster runtime at high error rates.
} \vspace{3pt} 
\centering      
\begin{tabular}{lrrr}
\hline
average iteration counts ~~~ & $p$=0.255 & ~~~~$p$=0.328 & ~~~~$p$=0.392 \\
\hline
AMBP$_4$, $\alpha_1=1.2$ & 3.28 & 10.54 & 4628.89 \\ 
AMBP$_4$, $\alpha_1=\text{func}(p)$  & 3.28 & 8.33 & 90.27 \\ 
MBP$_4$, $\alpha=1.2$ & 3.27 & 6.14 & 80.12 \\
\hline
\hline
GD AMBP$_2$, $\alpha_1=1.2$ & 3.30 & 11.98 & 5330.62 \\ 
GD AMBP$_2$, $\alpha_1=\text{func}(p)$  & 3.30 & 10.67 & 94.75 \\ 
GD MBP$_2$, $\alpha=1.2$ & 3.30 & 6.80 & 83.55 \\ 
\hline
\hline
GD Flip-BP$_2$ & 2.99 & 5.18 & 51.51 \\
\hline
\end{tabular} \label{tb:iter_GHP} 
\end{table}




\section*{Data Availability}
We provide a software implementation package available at ~ \url{https://github.com/kywukuo/qLDPC_dec}

\section*{Acknowledgements}
K.K. and Y.O. acknowledge support from EPSRC (Grant No.~EP/W028115/1). 
Y.O. also acknowledges support from the EPSRC funded QCI3 Hub under Grant No. EP/Z53318X/1.

\section*{Author contributions}
All authors wrote the manuscript. KY Kuo performed the numerics of the manuscript.

\section*{Competing interests}
The authors declare no competing interests.



\renewcommand*{\bibfont}{\footnotesize}

%

\foreach \x in {1,...,15} 
{%
\clearpage
\includepdf[pages={\x}]{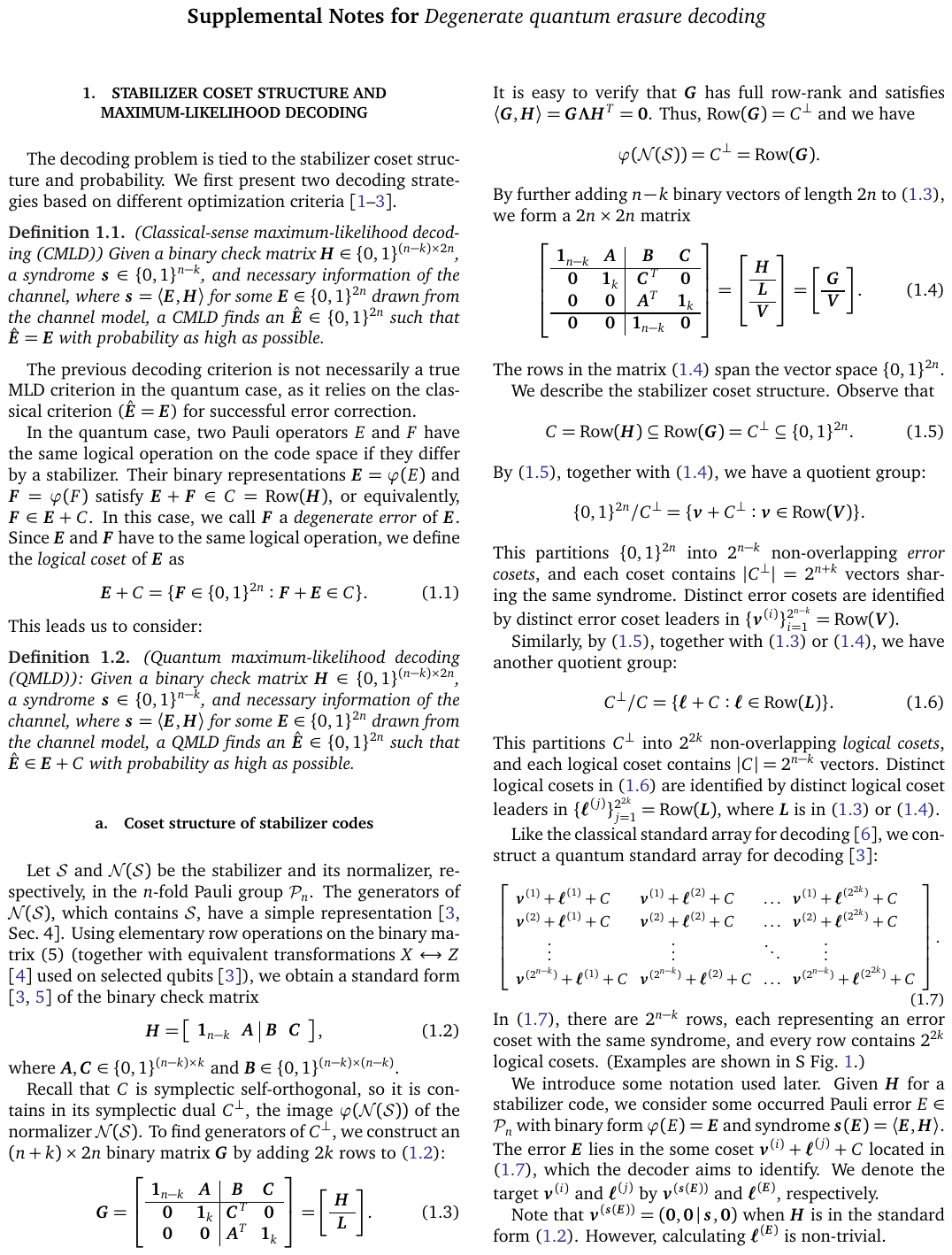} 
}

\end{document}